\documentclass{CSML}

\def\dOi{13(3:1)2017}
\lmcsheading%
{\dOi}
{1--31}
{}
{}
{May\phantom.~10, 2016}
{Jul.~\phantom06, 2017}
{}

\usepackage{hyperref}\hypersetup{hidelinks}
\usepackage{amssymb}
\usepackage{amsmath}
\usepackage{amsfonts}
\usepackage{tikz}
\usepackage{circuitikz}

\newcommand{\meet}{\wedge}
\newcommand{\join}{\vee}
\newcommand{\leqc}{\prec}

\theoremstyle{plain}\newtheorem*{M0*}{M.0}
\theoremstyle{plain}\newtheorem*{M1*}{M.1}
\theoremstyle{plain}\newtheorem*{M2*}{M.2}
\theoremstyle{plain}\newtheorem*{M3*}{M.3}
\theoremstyle{plain}\newtheorem*{M4*}{M.4}

\DeclareMathOperator{\free}{Free}
\DeclareMathOperator{\eval}{Eval}
\DeclareMathOperator{\dual}{op}

\newcommand{\C}{\mathcal{C}}
\newcommand{\Ell}{\mathcal{L}}
\newcommand{\K}{\mathcal{K}}
\renewcommand{\S}{\mathcal{S}}
\renewcommand{\L}{\mathcal{L}}
\newcommand{\F}{\mathcal{F}}
\newcommand{\E}{\mathcal{E}}
\newcommand{\W}{\mathcal{W}}
\newcommand{\plus}{+}
\newcommand{\proj}{\cdot}
\DeclareMathOperator{\fragility}{fragility}
\DeclareMathOperator{\resilience}{resilience}
\newcommand{\Rnneg}{\mathbb{R}^+}

\begin{document}

\title[Towards an Algebra for Cascade Effects]{Towards an Algebra for Cascade Effects}

\author[E.~M.~Adam]{Elie M. Adam\rsuper *}	%required
\address{Laboratory for Information and Decision Systems,
  Massachusetts Institute of Technology}
\email{\{eadam, dahleh, asuman\}@mit.edu}  %optional
\thanks{{\lsuper *}This research was partially supported by a Xerox Fellowship.}	%optional

\author[M.~A.~Dahleh]{Munther A. Dahleh}	%optional
\address{\vspace{-18 pt}}
%\email{dahleh@mit.edu}  %optional

\author[A.~Ozdaglar]{Asuman Ozdaglar}	%optional
\address{\vspace{-18 pt}}
%\email{asuman@mit.edu}  %optional

\ACMCCS{[{\bf Mathematics of
  computing}]: Discrete mathematics; [{\bf Theory of computation}]: Logic; [{\bf Applied computing}]: Engineering}

\begin{abstract}
\noindent We introduce a new class of (dynamical) systems that inherently capture cascading effects (viewed as consequential effects) and are naturally
amenable to combinations. We develop an axiomatic general theory around those systems, and guide the endeavor towards an understanding of cascading failure.
The theory evolves as an interplay of lattices and fixed points, and its results may be instantiated to commonly studied \emph{models} of cascade effects.

We characterize the systems through their fixed points, and equip them with two operators. We uncover properties of the operators, and express \emph{global} systems through combinations of \emph{local} systems. We enhance the theory with a notion of failure, and understand the class of shocks inducing a system to failure. We develop a notion of $\mu$-rank to capture the energy of a system, and understand the minimal amount of \emph{effort} required to fail a system, termed \emph{resilience}. We deduce a dual notion of \emph{fragility} and show that the combination of systems sets a limit on the amount of fragility inherited.
\end{abstract}

\maketitle

\section{Introduction}

Cascade effects refer to situations where the expected behavior governing a certain 
system appears to be \emph{enhanced} as this component is embedded into a greater system. The effects of change
in a subsystem may pass through \emph{interconnections} and enforce an indirect change on the state 
of any remote subsystem. As such effects are pervasive---appearing in various scenarios of ecological systems, 
communication infrastructures, financial networks, power grids and societal networks---there is an 
interest (and rather a need) to understand them. Models are continually proposed to capture 
instances of cascading behavior, yet the \emph{universal} properties of this phenomenon remain 
untouched. Our goal is to capture some essence of cascade effects, and develop an axiomatic theory around it.

A reflection on such a phenomenon reveals two informal aspects of it. The first aspect 
uncovers a notion of \emph{consequence} relation that seemingly drives the phenomenon. Capturing 
\emph{chains of events} seems to be inescapably necessary. The second aspect projects cascade 
effects onto a theory of subsystems, combinations and interaction. We should not expect any 
cascading behavior to occur in \emph{isolation}.

The line of research will be pursued within the context of systemic failure, and set along a guiding 
informal question. When handed a system of interlinked subsystems, when would a \emph{small} perturbation in some 
subsystems induce the system to failure?
The phenomenon of cascade effects (envisioned in this paper) restricts the possible systems to those satisfying posed axioms.
The analysis of cascade effects shall be perceived through an analysis on these systems.

We introduce a new class of (dynamical) systems that inherently capture cascading effects (viewed as \emph{consequential} effects) and are naturally amenable to combinations. We develop a general theory around those systems, and guide the endeavor towards an understanding of cascading failure. The theory evolves as an interplay of lattices and fixed points, and its results may be instantiated to commonly studied \emph{models} of cascade effects.

\subsection*{Our Systems}

The systems, in this introduction, will be motivated through an elementary example. This example is labeled M.0 and further referred to throughout the paper.

\begin{M0*} Let $G(V,A)$ be a digraph, and define $N(S) \subseteq V$ to be the set of nodes $j$ with $(i,j) \in A$ and $i \in S$. A vertex is of one of two colors, either black or white. The vertices are initially colored, and $X_0$ denotes the set of black colored nodes. The system evolves through discrete time to yield $X_1,X_2,\cdots$ sets of black colored nodes. Node $j$ is colored black at step $m+1$ if any of its neighbors $i$ with $j \in N(i)$ is black at step $m$. Once a node is black it remains black forever.
\end{M0*}
    
Our systems will consist of a collection of states along with internal dynamics. The collection of states is a finite set $P$. The dynamics dictate the evolution of the system through the states and are \emph{governed} by a class of maps $P\rightarrow P$. The state space in M.0 is the set $2^V$ where each $S \subseteq V$ identifies a subset of \emph{black} colored nodes; the dynamics are dictated by $g: X \mapsto X \cup N(X)$ as $X_{m+1} = g X_m$.

We intuitively consider some states to be \emph{worse} or \emph{less desirable} than others. The color \emph{black} may be undesirable in M.0, representing a \emph{failed} state of a node. State $S$ is then considered to be worse than state $T$ if it includes $T$. We formalize this notion by equipping $P$ with a partial order $\leq$. The order is only partial as not every pair of states may be comparable. It is natural to read $a \leq b$ in this paper as $b$ is a worse (or less desirable) state than $a$. The state space $2^V$ in M.0 is ordered by set inclusion $\subseteq$.

We expect two properties from the dynamics driving the systems. We require the dynamics to be \emph{progressive}. The system may only evolve to a state considered less desirable than its initial state. We also require \emph{undesirability} to be preserved during an evolution. The less desirable the initial state of a system is, the less desirable the final state (that the system evolves to) will be. We force each map $f: P \rightarrow P$ governing the dynamics to satisfy two axioms:
\begin{description}
 \item[A.1] If $a \in P$, then $a \leq fa$.
 \item[A.2] If $a, b \in P$ and $a \leq b$, then $fa \leq fb$.
\end{description}
The map $X \mapsto X \cup N(X)$ in M.0 satisfies both A.1 and A.2 as $S \subseteq S \cup N(S)$, and $S\cup N(S) \subseteq S'\cup N(S')$ if $S \subseteq S'$.

Our interest lies in the limiting outcome of the dynamics, and the understanding we wish to develop may be solely based on the \emph{asymptotic} behavior of the system. In M.0, we are interested in the set $X_m$ for $m$ large enough as a function of $X_0$. As $V$ is finite, it follows that $X_m = X_{|V|}$ for $m \geq |V|$. We are thus interested in the map $g^{|V|} : X_0 \mapsto X_{|V|}$. More generally, as iterative composition of a map satisfying A.1 and A.2 eventually yields idempotent maps, we equip the self-maps $f$ on $P$ with a third axiom:
\begin{description}
 \item[A.3] If $a \in P$, then $ffa = fa$.
\end{description}
Our class of interest is the (self-)maps (on $P$) satisfying the axioms A.1, A.2 and A.3. Each system will be identified with one such map. The system generated from an instance of M.0 corresponds to the map $X_0 \mapsto X_{|V|}$.

The axioms A.1, A.2 and A.3 naturally permeate a number of areas of mathematics and logic.
Within metamathematics and (universal) logic, Tarski introduced these three axioms 
(along with supplementary axioms) and launched his theory of consequence operator (see \cite{TAR1936} and \cite{TAR1956}). He aimed to 
provide a general characterization of the notion of deduction. As such, if $S$ represents a
set of statements taken to be true (i.e. premises), and $Cn(S)$ denotes the set of statements 
that can be \emph{deduced} to be true from $S$, then $Cn$ (as an operator) obeys A.1, A.2 and A.3.
Many familiar maps also adhere to the axioms. As examples, we may consider the function that maps %
 (i) a subset of a topological space to its topological closure,
 (ii) a subset of a vector space to its linear span,
 (iii) a subset of an algebra (e.g. group) to the subalgebra (e.g. subgroup) it generates,
 (iv) a subset of a Euclidean n-space to its convex hull.
Such functions may be referred to as \emph{closure operators} (see e.g. \cite{BIR1936}, \cite{BIR1967},  \cite{ORE1943} and \cite{WAR1942}), and are typically objects of study in \emph{universal algebra}.

\subsection*{Goal and Contribution of the Paper}
This paper has three goals. The first is to introduce and motivate the class of systems. The second is to present some properties of the systems, and develop preliminary tools for the analysis. The third is to construct a setup for cascading failure, and illustrate initial insight into the setup. The paper will not deliver an exhaustive exposition. It will introduce the concepts and augment them with enough results to allow further development.

We illustrate the contribution through M.0. We define $f$ and $g$ to be the systems derived from two instances $(V,A)$ and $(V,A')$ of M.0.

We establish that our systems are uniquely identified with their set of fixed points. We can reconstruct $f$ knowing only the sets $S$ containing $N(S)$ (i.e. the fixed points of $f$) with no further information on $(V,A)$. We further provide a complete characterization of the systems through the fixed points. The characterization yields a remarkable conceptual and analytical simplification in the study.

We equip the systems with a lattice structure, uncover operators ($\plus$ and $\proj$) and express \emph{complex} systems through formulas built from \emph{simpler} systems. The $\plus$ operator \emph{combines} the effect of systems, possibly derived from different models. The system $f\plus g$, as an example, is derived from $(V,A \cup A')$. The $\proj$ operator \emph{projects} systems onto each other allowing, for instance, the recovery of \emph{local} evolution rule. We fundamentally aim to extract properties of $f \plus g$ and $f \proj g$ through properties of $f$ and $g$ separately. We show that $\plus$ and $\proj$ lend themselve to well behaved operations when systems are represented through their fixed-points.

We realize the systems as interlinked components and formalize a notion of \emph{cascade effects}. Nodes in $V$ are identified with maps $e_1,\cdots, e_{|V|}$. The system $f\proj e_i$ then defines the evolution of the color of node $i$ as a function of the system state, and is identified with the set of nodes that \emph{reach} $i$ in $(V,A)$.

We draw a connection between shocks and systems, and enhance the theory with a notion of failure.
We show that minimal shocks (that fail a system $h$) exhibit a unique property that uncovers \emph{complement} subsystems in $h$, termed \emph{weaknesses}. A system is shown to be \emph{injectively} decomposed into its \emph{weaknesses}, and any weakness in $h \plus h'$ cannot result but from the combination of weaknesses in $h$ and $h'$.

We introduce a notion of $\mu$-rank of a system---akin to the (analytic) notion 
of a norm as used to capture the energy of a system---and
show that such a notion is unique should it adhere to natural principles.
The $\mu$-rank is tied to the number of connected components in $(V,A)$ when $A$ is symmetric.

We finally set to understand the minimal amount of \emph{effort} required to fail a system, termed \emph{resilience}. Weaknesses reveal a dual (equivalent) quantity, termed \emph{fragility}, and further puts resilience and $\mu$-rank on comparable grounds. The fragility is tied to the size of the largest connected component in $(V,A)$ when $A$ is symmetric. It is thus possible to formally define \emph{high ranked} systems that are not necessarily \emph{fragile}. The combination of systems sets a limit on the amount of fragility the new system inherits. Combining two subsystems cannot form a fragile system, unless one of the subsystems is initially fragile.

\subsection*{Outline of the Paper} 
Section 2 presents mathematical preliminaries. We characterize the systems in Section 3, and equip them with the operators in Section 4. We discuss component realization in Section 5, and derive properties of the systems lattice in Section 6.  We discuss cascade effects in Section 7, and provide connections to formal methods in Section 8. We consider cascading failure and resilience in Section 9, and conclude with some remarks in Section 10.

\section{Mathematical Preliminaries}
A partially ordered set or poset $(P,\leq)$ is a set $P$ equipped with a (binary) relation $\leq$ that is reflexive, antisymmetric and transitive. The element $b$ is said to cover $a$ denoted by $a \leqc b$ if $a \leq b$, $a \neq b$ and there is no $c$ distinct from $a$ and $b$ such that $a \leq c$ and $c \leq b$. A poset $P$ is graded if, and only if, it admits a rank function $\rho$ such that $\rho(a) = 0$ if $a$ is minimal and $\rho(a') = \rho(a) + 1$ if $a \leqc a'$. The poset $(P,\leq)$ is said to be a lattice if every pair of elements admits
a greatest lower bound (meet) and a least upper bound (join) in $P$. We define
the operators $\meet$ and $\join$ that sends a pair to their meet and join
respectively. The structures $(P,\leq)$ and $(P,\meet,\join)$ are
then isomorphic. A lattice is distributive if, and only if, $(a \join b) \meet c =
(a \meet c) \join (b \meet c)$ for all $a$, $b$ and $c$. The pair $(a,b)$ is said to be a modular pair if $c \join (a \meet b) = (c \join a) \meet b$ whenever $c \leq b$. A lattice is modular if all pairs are modular pairs. Finally, a \emph{finite} lattice is (upper) semimodular if, and only if, $a \join b$ covers both $a$ and $b$, whenever $a$ and $b$ cover $a \meet b$.

\subsection*{Notation} We denote $f(g(a))$ by $fga$, the composite $ff$ by $f^2$, and the inverse map of $f$ by $f^{-1}$. We also denote $f(i)$ by $f_i$ when convenient.

\section{The Class of Systems}\label{sec:class}

The state space is taken to be a \emph{finite} lattice $(P,\leq)$. We consider in this paper only posets $(P,\leq)$ that are lattices, as opposed to arbitrary posets. It is natural to read $a \leq b$ in this paper as $b$ is a worse (or less desirable) state than $a$. The meet (glb) and join (lub) of $a$ and $b$ will be denoted by $a\meet b$ and $a \join b$ respectively. A minimum and maximum element exist in $P$ (by finiteness) and will be denoted by $\check{p}$ and $\hat{p}$ respectively.

A system is taken to be a map $f: P \rightarrow P$ satisfying:
\begin{description}
 \item[A.1] If $a \in P$, then $a \leq fa$.
 \item[A.2] If $a, b \in P$ and $a \leq b$, then $fa \leq fb$.
 \item[A.3] If $a \in P$, then $ffa = fa$.
\end{description}
The set of such maps is denoted by $\L_P$ or simply by $\L$ when $P$ is irrelevant to the context. This set is necessarily finite as $P$ is finite.

\subsubsection*{Note on Finiteness} Finiteness is not essential to the development
in the paper; completeness can be used to replace finiteness when needed. We restrict the exposition in this paper to finite cases to ease non-necessary details. As every finite lattice is complete, we will make no mention of completeness throughout.

\subsection{Models and Examples}
The axioms A.1 and A.2 hold for typical ``models'' adopted for cascade effects. We present three models (in addition to M.0 provided in Section 1) supported on the Boolean lattice, two of which---M.1 and M.3---are standard examples (see \cite{GRA1978}, \cite{KLE2007} and \cite{MOR2000}). It can be helpful to identify a set $2^S$ with the set of all $black$ and $white$ colorings on the objects of $S$. A subset of $S$ then denotes the objects colored $black$. The model M.1 generalizes M.0 by assigning \emph{thresholds} to nodes in the graph. Node $i$ is colored $black$ when the number of neighbors colored $black$ surpasses its threshold. The model M.2 is \emph{noncomparable} to M.0 and M.1, and the model M.3 generalizes all of M.0, M.1 and M.2.

\begin{M1*} Given a digraph over a set $S$ or equivalently a map $N : S \rightarrow 2^S$, a map $k:S \rightarrow \mathbb{N}$ and a subset $X_0$ of $S$, let $X_1,X_2,\cdots$ be subsets of $S$ recursively defined such that $i \in X_{m+1}$ if, and only
  if, either $|N_i \cap X_m| \geq k_i$ or $i\in X_m$.
\end{M1*}
\begin{M2*} Given a collection $\mathcal{C} \subseteq 2^S$ for some set $S$, a map $k:\mathcal{C}\rightarrow \mathbb{N}$ and a subset $X_0$ of $S$, let $X_1,X_2,\cdots$ be
  subsets of $S$ recursively defined such that $i \in X_{m+1}$ if, and only
  if, either there is a $C\in \mathcal{C}$ containing $i$ such that $|C \cap X_m| \geq k_c$ or $i\in X_m$.
\end{M2*}
\begin{M3*} Given a set $S$, a collection of monotone maps $\phi_i$ (one for each $i \in S$) from $2^S$
  into $\{0,1\}$ (with $0<1$) and a subset $X_0$ of $S$, let $X_1,X_2,\cdots$ be
  subsets of $S$ recursively defined such that $i \in X_{m+1}$ if, and only
  if, either $\phi_i (X_m) = 1$ or $i\in X_m$.
\end{M3*}
We necessarily have $X_{|S|} = X_{|S|+1}$ in the three cases above, and the map $X_0 \mapsto X_{|S|}$ is then in $\Ell_{2^S}$. The dynamics depicted above may be captured in a more general form.
\begin{M4*} Given a finite lattice $L$, an order-preserving map $h: L \rightarrow L$, and $x_0 \in L$, let $x_1,x_2,\cdots \in L$ be recursively defined such that $x_{m+1} = x_m \join h(x_m)$.  
\end{M4*}
We have $x_{|L|} = x_{|L|+1}$ and the map $x_0 \mapsto x_{|L|}$ is in $\Ell_L$.

The axioms allow greater variability if the state space is modified or augmented accordingly. Nevertheless, this paper is only concerned with systems of the above form.

\subsubsection*{Note on Realization} Modifications of instances of M.i (e.g. altering values of $k$ in M.1) may not alter the system function. As the interest lies in understanding \emph{universal} properties of final evolution states, the analysis performed should be invariant under such modifications. However, analyzing the systems directly through their \emph{form} (as specified through M.0, M.1, M.2 and M.3) is bound to rely heavily on the representation used. Introducing the axioms and formalism enables an understanding of systems that is independent of their representation. It is then a separate question as to whether or not a system may be realized through some form, or whether or not restrictions on form translate into interesting properties on systems. Not all systems supported on the Boolean lattice can be realized through the form M.0, M.1 or M.2. However, every system in $\Ell_{2^S}$ may be realized through the form M.3.  Indeed, if $f \in \Ell_{2^S}$, then for every $i \in S$ define $\phi_i : 2^S \rightarrow \{0,1\}$ where $\phi_i(a) = 1$ if, and only if, $i \in f(a)$. The map $\phi_i$ is monotone as $f$ satisfies A.2.  Realization is further briefly discussed in Section~\ref{sec:comp}.

\subsection{Context, Interpretation and More Examples.}

A more \emph{realistic} interpretation of the models M.i comes from a more realistic interpretation of the state space.  This work began as an endeavor to understand the mathematical structure underlying models of diffusion of behavior commonly studied in the social sciences.  The setup there consists of a population of interacting agents.  In a societal setting, the agents may refer to individuals.  The interaction of the agents affect their behaviors or opinions.  The goal is to understand the spread of a certain behavior among agents given certain interaction patterns.  Threshold models of behaviors (captured by M.0, M.1, M.2 and M.3) have appeared in the work of Granovetter \cite{GRA1978}, and more recently in \cite{MOR2000}.  Such models are key models in the literature, and have been later considered by computer scientists, see. e.g., \cite{KLE2007} for an overview.

The model described by M.1 is known as the linear threshold model.  An individual adopts a behavior, and does not change it thereafter, if at least a certain number of its neighbors adopts that behavior.  Various variations can also be defined, see e.g.\ M.2 and M.3, and again \cite{KLE2007} for an overview.  The cascading intuition in all the variations however remains unchanged.  These models can generally be motivated through a game theoretic setup.  We will not be discussing such setups in this paper.  The \emph{no-recovery} aspect of the models considered may be further relaxed by introducing appropriate time stamps.  One such connection is described in \cite{KLE2007}. We are however interest in the instances where no-recovery occurs.
 
The models may also be given an interpretation in epidemiology.  Every agent may either be healthy or infected.  Interaction with an infected individual causes infections.  This is in direct resemblance to M.0.  Stochastics can also be added, either for a realistic approach or often for tractability.  There is also a vast literature on processes over graphs, see e.g., \cite{DUR1997} and \cite{NEW2010}.  Our aim is to capture the consequential effects that are induced by the interaction of several entities.  We thus leave out any stochatics for the moment; they may be added later with technical work.

On a different end, inspired by cascading failure in electrical grids, consider the following simple resistive circuit.  The intent is to guide the reader into a more realistic direction.
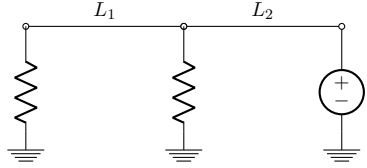
\begin{center}
  \begin{circuitikz}[scale=0.7, transform shape]\draw
    (0,1) to [short, o-o, l=$L_1$] (3,1) to [short, o-o, l=$L_2$] (6,1)
    (6,1) -- (6,0.5)  to[american voltage source] (6,-1) node[ground]{};
    \draw (0,1) -- (0,0.5) to[R] (0,-1) node[ground]{};
    \draw (3,1) -- (3,0.5) to[R] (3,-1) node[ground]{};
\end{circuitikz}
\end{center}
If line $L_2$ is disconnected from the voltage source, then line $L_1$ will also be disconnected from the source.  Indeed, the current passing through $L_1$ has to pass through $L_2$.  The converse is, of course, not true.  This interdependence between $L_1$ and $L_2$ is easily captured by a system in $\Ell_{2^{\{L_1,L_2\}}}$. More general dependencies (notably failures caused by a redistribution of currents) can be captured, and concurrency can be taken care of by going to power sets.  Indeed, M.0 also captures general reachability problems, where a node depicts an element of the state space.  Specifically, let $S$ be a set of states of some system, and consider a reflexive and transitive relation $\rightarrow$ such that $a \rightarrow b$ means that state $b$ is reachable from state $a$.  The map $2^S \rightarrow 2^S$ where $A \mapsto \{b : a \rightarrow b \text{ for some } a \in A\}$ satsfies A.1, A.2, and A.3 when $2^S$ is ordered by inclusion.

This work abstracts out the essential properties that gives rise to these situations.  The model M.3 depicts the most general form over the boolean lattice.  In M.3, the set $S$ can be interpreted to contain $n$ events, and an element of $2^S$ then depicts which events have occured.  A system is then interpreted as a collection of (monotone) implications:  if such and such event occurs, then such event occurs.  The more general model M.4 will be evoked in Section \ref{sec:formal}, while treating connections to formal methods and semantics of programming languages.

\subsubsection*{On Closure Operators} As mentioned in the introduction, the maps satisfying A.1, A.2 and A.3 are often known as closure operators.  On one end, they appeared in the work of Tarski (see e.g., \cite{TAR1936} and \cite{TAR1956}). On another end, they appeared in the work of Birkhoff, Ore and Ward (see e.g., \cite{BIR1936}, \cite{ORE1943} and \cite{WAR1942}, respectively).  The first origin reflects the consequential relation in the effects considered.  The second origin reflects the theory of interaction of multiple systems.  Closure operators appear as early as \cite{MOO1911}.  They are intimately related to Moore families or closure systems (i.e., collection of subsets of $S$ containing $S$ and closed under intersection) and also to Galois connections (see e.g., \cite{BIR1967} Ch. V and \cite{EVE1944}).
Every closure operator corresponds to a Moore family (see e.g., Subsection \ref{subsec:fp}). This connection will be extensively used throughout the paper.  Most of the properties derived in Sections \ref{sec:class} and \ref{sec:lattice} can be seen to appear in the literature (see e.g. \cite{BIR1967} Ch.\ V and \cite{CAS2003} for a recent survey).   They are very elementary, and will be easily and naturally rederived whenever needed.
Furthermore, every Galois connection induces one closure operator, and every closure operator arises from at least one Galois connection.  Galois connection will be briefly discussed in Section \ref{sec:eval}.  They will not however play a major explicit role in this paper.

\subsection{The Fixed Points of the Systems}\label{subsec:fp}

As each map in $\L$ sends each state to a respective fixed point, a
grounded understanding of a system advocates an understanding of its
fixed points. We develop such an understanding in this subsection, and
characterize the systems through their fixed points. Let $\Phi$ be the
map $f \mapsto \{a : fa=a\}$ that sends a system to its set of fixed
points.\newpage

\begin{prop}
 If $f\neq g$ then $\Phi f \neq \Phi g$.
\end{prop}
\begin{proof}
 If $\Phi f = \Phi g$, then $ga \leq gfa = fa$ and $fa \leq fga = ga$ for each $a$. Therefore $f = g$.
\end{proof}
It is obvious that each state is mapped to a fixed point; it is less obvious that, knowing only the fixed points, the system can be reconstructed uniquely. It seems plausible then to directly define systems via their fixed point, yet doing so inherently supposes an understanding of the image set of $\Phi$.
\begin{prop}\label{pro:maxp}
 If $f \in \L_P$, then $\hat{p} \in \Phi f$.
\end{prop}
\begin{proof}
 Trivially $\hat{p} \leq f\hat{p} \leq \hat{p}$.
\end{proof}
Furthermore,
\begin{prop}\label{pro:closedmeet}
 If $a,b\in \Phi f$, then $a \meet b \in \Phi f$.
\end{prop}
\begin{proof}
 It follows from A.2 that $f(a\meet b) \leq fa$ and $f(a\meet b) \leq fb$. If $a,b \in \Phi f$, then $f(a\meet b) \leq fa \meet fb = a\meet b$. The result follows as $a\meet b \leq f(a \meet b)$.    
\end{proof}
In fact, the properties in Propositions \ref{pro:maxp} and \ref{pro:closedmeet} fully characterize the image set of $\Phi$.
\begin{prop}
 If $S \subseteq P$ is closed under $\meet$ and contains
 $\hat{p}$, then $\Phi f=S$ for some $f \in \L_P$.
\end{prop}
\begin{proof}
 Construct $f : a \mapsto \inf\{b \in S : a \leq b\}$. Such a function is well defined and satisfies A.1, A.2 and A.3.
\end{proof}

It follows from Propositions \ref{pro:maxp} and \ref{pro:closedmeet} that $\Phi f$ forms a lattice under the induced order $\leq$. This conclusion coincides with that of Tarski's fixed point theorem (see \cite{TAR1955}). However, one additional structure is gained over arbitrary order-preserving maps. Indeed, the meet operation of the lattice $(\Phi f,\leq)$ coincides with that of the lattice $(P,\leq)$.

\begin{exa} Let $f : 2^V \rightarrow 2^V$ be the system derived from an instance $(V,A)$ of M.0. The fixed points of $f$ are the sets $S \subseteq V$ such that $S \supseteq N(S)$. If $S$ and $T$ are fixed points of $f$, then $S \cap T$ is a fixed point of $f$. Indeed, the set $S\cap T$ contains $N(S \cap T)$. The map $f$ sends each set $T$ to the intersection of all sets $S \supseteq T \cup N(S)$. Although every collection $C$ of sets in $2^V$ closed under $\cap$ and containing $V$ can form a system, it will not always be possible to find a digraph where $C$ coincides with the sets $S \supseteq N(S)$. The model M.0 is not \emph{complex} enough to capture all possible systems. 
\end{exa}
  
The space $\L$ is thus far only a set, with no further mathematical structure. The theory becomes lively when elements of $\L$ become \emph{related}.

\subsection{Overview Through an Example}\label{sec:running}

We illustrate some main ideas of the paper through an elementary example.  The example will run throughout the paper, revisited in each section to illustrate its corresponding notions and results. The example we consider is the following (undirected) instance of M.1:

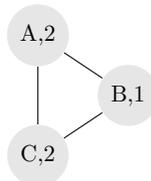
\begin{center}
\begin{tikzpicture}
  [scale=.4,auto=center,every node/.style={circle,fill=black!10!white,scale=0.8}]
  \node (n1) at (1,10) {A,2};
  \node (n2) at (4,8)  {B,1};
  \node (n3) at (1,6)  {C,2};

    \foreach \from/\to in {n1/n2,n2/n3,n1/n3}
    \draw (\from) -- (\to);
\end{tikzpicture}
\end{center}

The nodes are labeled $A$, $B$ and $C$. Each node $I$ is tagged with an integer $k_I$ that denotes a \emph{threshold}. Each node can then be in either one of two colors: \emph{black} or \emph{white}. Node $I$ is colored black (and stays black forever) when at least $k_I$ neighbors are \emph{black}. In our example, node $A$ (resp. $C$) is colored black when both $B$ and $C$ (resp. $A$) are black. Node $B$ is colored black when either $A$ or $C$ are black. A node remains white otherwise.

The set underlying the state space is the set of possible colorings of nodes. Each coloring may be identified with a subset of $\{A,B,C\}$ containing the black colored nodes. The state space will then be identified with $2^{\bf{3}}$, the set of all subsets of $\{A,B,C\}$. The set $2^{\bf{3}}$ admits a natural ordering by inclusion ($\subseteq$) that turns it into a lattice. It may then be represented through a \emph{Hasse diagram} as:

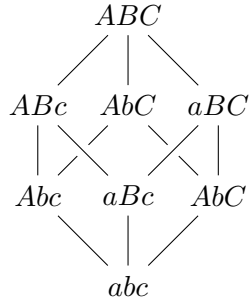
\begin{center}
\begin{tikzpicture}[scale=0.8]
  \node (max) at (0,3) {$ABC$};
  \node (a) at (1.5,1.5) {$aBC$};
  \node (b) at (0,1.5) {$AbC$};
  \node (c) at (-1.5,1.5) {$ABc$};
  \node (d) at (1.5,0) {$AbC$};
  \node (e) at (0,0) {$aBc$};
  \node (f) at (-1.5,0) {$Abc$};
  \node (min) at (0,-1.5) {$abc$};
  \draw (min) -- (d) -- (a) -- (max) -- (b) -- (f)
  (e) -- (min) -- (f) -- (c) -- (max)
  (d) -- (b);
  \draw[preaction={draw=white, -,line width=6pt}] (a) -- (e) -- (c);
\end{tikzpicture}
\end{center}

{\bf Notation:} We denote subsets of $\{A,B,C\}$ as strings of letters. Elements in the set are written in uppercase, while elements not in the set are written in lowercase. Thus $aBC$, $Abc$ and $abc$ denote $\{B,C\}$, $\{A\}$ and $\{\}$ respectively. The string $AC$ (with $b/B$ absent) denotes both $AbC$ and $ABC$.

The system derived from our example is the map $f : 2^{\bf 3} \rightarrow 2^{\bf 3}$ satisfying A.1, A.2 and A.3 such that $A \mapsto ABC$, $C \mapsto ABC$ and all remaining states are left unchanged. The fixed points of $f$ yield the following representation.

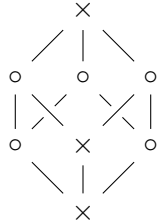
\begin{center}
\begin{tikzpicture}[scale=0.6]
  \node (max) at (0,3) {$\times$};
  \node (a) at (1.5,1.5) {$\circ$};
  \node (b) at (0,1.5) {$\circ$};
  \node (c) at (-1.5,1.5) {$\circ$};
  \node (d) at (1.5,0) {$\circ$};
  \node (e) at (0,0) {$\times$};
  \node (f) at (-1.5,0) {$\circ$};
  \node (min) at (0,-1.5) {$\times$};
  \draw (min) -- (d) -- (a) -- (max) -- (b) -- (f)
  (e) -- (min) -- (f) -- (c) -- (max)
  (d) -- (b);
  \draw[preaction={draw=white, -,line width=6pt}] (a) -- (e) -- (c);
\end{tikzpicture}
\end{center}

We indicate, on the diagram, a fixed point by $\times$ and a non-fixed point by $\circ$.

\subsection{On the System Maps and their Interaction}  \label{ap:motivation}
As mentioned in the introduction, the systems of interest consist of a collection of states along with internal dynamics. The collection of states is a finite set $P$. The dynamics dictate the evolution of the system through the states and are \emph{governed} by a class $\K$ of maps $P\rightarrow P$. The class $\K$ is closed under composition, contains the identity map and satisfies:
\begin{description}
 \item[P.1] If $a \neq b$ and $fa=b$ for some $f \in \K$, then $gb \neq a$ for every $g \in \K$.
 \item[P.2] If $gfa=b$ for some $f, g \in \K$, then $hga=b$ for some $h\in \K$.
\end{description}
The principles P.1 and P.2 naturally induce a partial order $\leq$ on the set $P$. The principles P.1 and P.2 further force the functions to be \emph{well adapted} to this order.
\begin{prop}
There exists a partial order $\leq$ on $P$ such that for each $f \in \K$:
\begin{description}
\item[A.1] If $a \in P$, then $a \leq fa$.
\item[A.2] If $a, b \in P$ and $a \leq b$, then $fa \leq fb$.
\end{description}
\end{prop}
\begin{proof}
 Define a relation $\leq$ on $P$ such that $a \leq b$ if, and only if, $b = fa$ for some $f\in \K$. The relation $\leq$ is reflexive and transitive as $\K$ is closed under composition and contains the identity map, respectively. Both antisymmetry and A.1 follow from P.1. Finally, if $a \leq b$, then $b = ga$ for some $g$. It then follows by P.2 that $fb = fga = hfa$ for some $h$. Therefore, $fa \leq fb$.
\end{proof}
We only alluded that the maps in $\K$ will govern our dynamics. No law of interaction is yet specified as to \emph{how} the maps will govern the dynamics. As the state space is finite, the interaction may be motivated by iterative (functional) composition. For some map $\phi : \mathbb{N} \rightarrow \K$, the system starts in a state $a_0$ and evolves through $a_1,a_2,\cdots$ with $a_{i+1} = \phi_i a_i$. We reveal properties of such an interaction.

Let $\phi : \mathbb{N} \rightarrow \S \subseteq \K$ be a surjective map, and define a map $F_i$ recursively as $F_1 = \phi_1$ and $F_{i+1} = \phi_{i+1}F_i$.
\begin{prop}
 For some $M$, we have $F_m = F_M$ for $m\geq M$.
\end{prop}
\begin{proof}
It follows from A.1 that $F_1a \leq F_2a \leq \cdots$. The result then follows from finiteness of $P$. 
\end{proof}
\begin{prop} \label{pro:idempotent}
The map $F_M$ is idempotent if $\phi^{-1}f$ is a non-finite set for each $f\in \S$.
\end{prop}
\begin{proof}
If $\phi^{-1}f$ is non-finite, then $fF = F$. If $\phi^{-1}f$ is non-finite for all $f \in \S$, then $FF = F$ as $F$ is the finite composition of maps in $\S$.
\end{proof}

Let $\psi : \mathbb{N} \rightarrow \S$ be another surjective map, and define a map $G_i$ recursively as $G_1 = \phi_1$ and $G_{i+1} = \phi_{i+1}G_i$. For some $N$, we necessarily get $G_{N} = G_n$ for $n \geq N$.
\begin{prop} \label{pro:equalmaps}
It follows that $F_M = G_N$, if $\phi^{-1}f$ and $\psi^{-1}f$ are non-finite sets for each $f\in \S$.
\end{prop}
\begin{proof} 
Define $F = F_M$ and $G = G_N$. As $F$ and $G$ are idempotent, then $FG = G$ and $GF = F$. Therefore $Fa \leq FGa = Ga$ and $Ga \leq GFa = Fa$.
\end{proof}

The maps governing the dynamics are to be considered as \emph{intrinsic mechanism wired into} the system. The effect of each map should not die out along the evolution of the system, but should rather keep on resurging. Such a consideration hints to an interaction insisting each map to be applied infinitely many times. There is immense variability in the order of application. However, we only want to care about the limiting outcome of the dynamics. By Proposition \ref{pro:equalmaps}, such a variability would then make no difference from our standpoint. We further know, through Proposition \ref{pro:idempotent}, that iterative composition in this setting cannot lead but to idempotent maps. We then impose---with no loss in generality---a third principle (P.3) on $\K$ to contain only idempotent maps. This principle gives rise to a third axiom.

\begin{description}
\item[A.3] For $a \in P$, $ffa = fa$.
\end{description}

We define $\L_P$ to be the set of maps satisfying A.1, A.2 and A.3. The set $\L_P$ is closed under composition and contains each element of $\K$ with P.3 imposed, including the identity map. Furthermore, the principles P.1, P.2 and P.3 remain satisfied if $\K$ is replaced by $\L_P$. We will then extend $\K$ to be equal to $\Ell_P$. This extension offers greater variability in dynamics, and there is no particular reason to consider any different set. We further consider only posets $(P,\leq)$ that are lattices, as opposed to arbitrary posets.

\section{The Lattice of Systems}\label{sec:lattice}

The theory of cascade effects presented in this paper is foremost a theory of combinations and interconnections. As such, functions shall be treated in relation to each other. The notion of desirability on states introduced by the partial order translates to a notion of desirability on systems. We envision that systems combined together should form less desirable systems, i.e. systems that more likely to evolve to less desirable states. Defining an order on the maps is natural to formalize such an intuition. We define the relation $\leq$ on $\L$, where $f \leq g$ if, and only if, $fa \leq ga$ for each $a$. 

\begin{prop}\label{pro:Llattice}
The relation $\leq$ is a partial order on $\Ell$, and the poset $(\L,\leq)$ is a lattice.
\end{prop} 
\begin{proof}
 The reflexivity, antisymmetry and transitivity properties of $\leq$ follow easily from A.1 and A.2. If $f,g \in \L$, then define $h : a \mapsto fa \meet ga$. It can be checked that $h \in \L$. Let $h'$ be any lower bound of $f$ and $g$, then $h'a \leq fa$ and $h'a \leq ga$. Therefore $h'a \leq fa \meet ga = ha$, and so every pair in $\L$ admits a greatest lower bound in $\L$. Furthermore, the map $a \mapsto \hat{p}$ is a maximal element in $\L$. The set of upper bounds of $f$ and $g$ in $\L$ is then non-empty, and necessarily contains a least element by finiteness. Every pair in $\L$ then also admits a least upper bound in $\L$.
\end{proof}

We may then deduce join and meet operations denoted by $\plus$ (combine) and $\proj$ (project) respectively.
The meet of a pair of systems was derived in the proof of Proposition \ref{pro:Llattice}.
\begin{prop}
 If $f,g \in \L$, then $f\proj g$ is $a \mapsto fa \meet ga$.
\end{prop}
On a dual end,
\begin{prop} \label{pro:join}
 If $f,g \in \L$, then $f\plus g$ is the least fixed point of the map $h \mapsto (fg)h(fg)$. As $P$ is finite, it follows that $f \plus g = (fg)^{|P|}$.
\end{prop}
\begin{proof}
Define $h_0 = (fg)^{|P|}$. Since the map $fg$ satisfies A.1 and A.2, then $h_0$ satisfies A.1 and A.2. Furthermore, iterative composition yields $(fg)^{|P|+1} = (fg)^{|P|}$. Then $h_0$ is idempotent i.e. satisfies A.3. The map $h_0$ is then a fixed point of $h \mapsto (fg)h(fg)$. Moreover, every upperbound on $f$ and $g$ is a fixed point of $h \mapsto (fg)h(fg)$. Let $h'$ be such an upperbound, then $fh'=h'$ and $gh' = h'$. It follows that $h_0h' = h'$ i.e. $h_0 \leq h'$. 
\end{proof}
The lattice $\L_P$ has a minimum and a maximum as it is finite. The minimum element (denoted by $0$ or $0_p$) corresponds to the identity map $a \mapsto a$. The maximum (denoted by $1$ or $1_p$) corresponds to $a \mapsto \hat{p}$.

\subsection{Interpretation and Examples}

The $\plus$ operator yields the most desirable system incorporating the effect of both of its operands. The $\proj$ operator dually yields the least desirable system whose effects are contained within both of its operands. Their use and significance is partially illustrated through the following six examples.

\subsubsection*{Example 0. Intuitive interpretation of the $\plus$ operator}  The $\plus$ operation combines the rules of the systems.  If each of $f$ and $g$ is seen to be described by a set of monotone deduction rules, then $f+g$ is the system that is obtained from the union of these sets of rules. The intuitive picture of combining rules may also found in the characterization $f \plus g = (fg)^{|P|}$.  Both rules of $f$ and $g$ are iteratively applied on an initial state to yield a final state. Furthermore, the order of composition does not affect the final state, as long as each system is applied enough times.  This insight follows from the interaction of A.1 and A.2, and is made formal in Subsection \ref{ap:motivation}. 

In a societal setting, each agent' state is governed by a set of local rules.  Every such set only affects the state of its corresponding agent.  The aggregate (via $\plus$) of all the local rules then defines the whole system.  It allows for an interaction between the rules, and makes way for cascade effects to emerge.  In the context of failures in infrastructure, the $\plus$ operator enables adding new conditions for failure/disconnections in the system.   This direction of aggregating local rules is further pursued in Section \ref{sec:comp} on component realization.  The definition of cascade effects is further expounded in Section \ref{sec:eval}.   The five examples to follow also provide additional insight.

\subsubsection*{Example 1. Overview on M.0} Let $f$ and $f'$ be systems derived from instances $(V,A)$ and $(V,A')$ of M.0. If $A' \subseteq A$, then $f' \leq f$. If $A'$ and $A$ are non-comparable, an inequality may still hold as different digraphs may give rise to the same system. The system $f \plus f'$ is the system derived from $(V, A\cup A')$. The system $f\proj f'$ is, however, not necessarily derived from $(V, A\cap A')$. If $(V,A)$ is a directed cycle and $(V,A')$ is the same cycle with the arcs reversed, then $f = f'$ while $(V, A\cap A')$ is the empty graph and yields the $0$ system.

\subsubsection*{Example 2. Combining Update Rules} Given a set $S$, consider a subset $N_i \subseteq S$ and an integer $k_i$ for each $i \in S$. Construct a map $f_i$ that maps $X$ to $X\cup \{i\}$ if $|X\cap N_i| \geq k_i$ and to $X$ otherwise. Finally, define the map $f = f_1 \plus \cdots \plus f_n$. The map $f$ can be realized by an instance of M.1, and each of the $f_i$ corresponds to a \emph{local evolution rule}.

\subsubsection*{Example 3. Recovering Update Rules} Given the setting of the previous example, define the map $e_i : X \mapsto X \cup \{i\}$. This map enables the extraction of a \emph{local} evolution rule. Indeed, $i \in (f \proj e_i) X_0$ if, and only if, $i \in f X_0$. However, if $j \neq i$, then $j \in (f \proj e_i) X_0$ if, and only if, $j \in X_0$. It will later be proved that $f = f\proj e_1 \plus \cdots \plus f\proj e_n$. The system $f$ can be realized as a combination of evolution rules, each governing the behavior of only one element of $S$.

\subsubsection*{Example 4. An Instance of Boolean Systems}Consider the following two instances of M.4, where $L$ is the Boolean lattice. Iteration indices are dropped in the notation.
\begin{align*}
x_1 &:= x_1 \join (x_2 \meet x_3) & x_1 &= x_1\nonumber\\
x_2 &:= x_2 \join x_3 & x_2 &= x_2 \join x_3\nonumber\\
x_3 &:= x_3 & x_3 &= x_3 \join (x_1 \meet x_2)\nonumber
\end{align*}
Let $f$ and $g$ denote the system maps generated by the right and left instances. The maps $f\plus g$ (left) and $f\proj g$ (right) can then be realized as:
\begin{align*}
x_1 &:= x_1 \join (x_2 \meet x_3) &\quad x_1 &= x_1\nonumber\\
x_2 &:= x_2 \join x_3 &\quad x_2 &= x_2 \join x_3\nonumber\\
x_3 &:= x_3 \join (x_1 \meet x_2) &\quad x_3 &= x_3\nonumber
\end{align*}
The map $f \proj g$ is the identity map.

\subsubsection*{Example 5. Closure under Meet and Join} If $f$ and $g$ are derived from instances of M.1, then neither $f \plus g$ nor $f \proj g$ are guaranteed to be realizable as instances of M.1. If they are derived from instances of M.2, then only $f \plus g$ is necessarily realizable as an instance of M.2. As all systems (over the Boolean lattice) can be realized as instances of M.3, both $f \plus g$ and $f \proj g$ can always be realized as instances of M.3. 

As an example, we consider the case of M.2.  If $(\mathcal{C}_f,k_f)$ and $(\mathcal{C}_g,k_g)$ are realizations of $f$ and $g$ as M.2, then $(\mathcal{C}_f\cup \mathcal{C}_g,k)$ is a realization of $f \plus g$, with $k$ being $k_f$ on $\mathcal{C}_f$ and $k_g$ on $\mathcal{C}_g$.  However, let $S= \{a,b,c\}$ be a set, and consider $\C_f = \{\{a,b\}\}$ with $k_f = 1$, and $\C_g = \{\{b,c\}\}$ with $k_g = 1$.  The set $\{a,c\}$ is not a fixed-point of $f \proj g$.  Thus, if a realization $(\C_{f\proj g},k)$ of $f \proj g$ is possible, then $\{a,b,c\}\in\C_{f\proj g}$ with $k \leq 2$. However, both $\{a,b\}$ and $\{b,c\}$ are fixed-points of $f\proj g$, contradicting such a realization.

\subsection{Effect of the Operators on Fixed Points}

The fixed point characterization uncovered thus far is independent of the order on $\L_P$. The map $\Phi : f \mapsto \{a: fa = a\}$ is also well behaved with respect to the $\plus$ and $\proj$ operations.  For $S, T \subseteq P$, we define their set meet $S \meet T$ to be $\{a\meet b : a\in S \text{ and } b \in T\}$.
\begin{prop}
 If $f,g \in \L$, then $\Phi (f\plus g) = \Phi f \cap \Phi g$ and $\Phi (f \proj g) = \Phi f \meet \Phi g$.
\end{prop}
\begin{proof}
  If $a \in \Phi f \cap \Phi g$, then $ga \in \Phi f$. As $fga = a$, it follows that $a \in \Phi(f\plus g)$. Conversely, as $(f+g)g = (f+g)$, if $(f+g)a = a$, then $(f+g)ga = a$ and so $ga = a$. By symmetry, if $(f+g)a=a$, then $fa=a$. Thus if  $a \in \Phi (f\plus g)$, then $a\in \Phi f \cap \Phi g$. Furthermore, $(f\proj g)a = a$ if, and only if, $fa \meet ga = a$ and the result $\Phi (f \proj g) = \Phi f \meet \Phi g$ follows.
\end{proof}

Combination and projection lend themselves to simple operations when the maps are viewed as a collection of fixed points. Working directly in $\Phi \L$ will yield a remarkable conceptual simplification.

\subsection{Summary on Fixed Points: The Isomorphism Theorem}

Let $\F$ be the collection of all $S \subseteq P$ such that $\hat{p} \in S$ and $a\meet b \in S$ if $a, b \in S$. Ordering $\F$ by reverse inclusion $\supseteq$ equips it with a lattice structure. The join and meet of $S$ and $T$ in $\F$ are, respectively, set intersection $S \cap T$ and set meet $S \meet T = \{a\meet b : a\in S \text{ and } b \in T\}$. The set $S \meet T$ may also be obtained by taking the union of $S$ and $T$ and closing the set under $\meet$.

\begin{thm}\label{thm:iso}
 The map $\Phi : f \mapsto \{a : fa=a\}$ defines an isomorphism between $(\L,\leq,\plus,\proj)$ and $(\F,\supseteq,\cap,\meet)$. \qed
\end{thm}

Such a result is well known in the study of \emph{closure operators}, and is relatively simple. We refer the reader, for instance, to \cite{BIR1936}, \cite{ORE1943} and \cite{WAR1942} for pieces of this theorem, and to \cite{BIR1967} Ch V and \cite{CAS2003} for a broader overview, more insight and references.  Nevertheless, the implications of it on the theory at hand can be remarkable. Our systems will be interchangeably used as both maps and subsets of $P$. The isomorphism enables a conceptual simplification, that enables emerging objects to be interpreted as systems exhibiting \emph{cascade effects}.

\subsection{Overview Through An Example (Continued)}

We continue the running example.  Our example is realized as a \emph{combination} of three evolution rules: one pertaining to each node. For instance, the rule of node $A$ may be realized as:
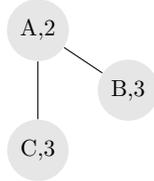
\begin{center}
\begin{tikzpicture}
  [scale=.4,auto=center,every node/.style={circle,fill=black!10!white,scale=0.8}]
  \node (n1) at (1,10) {A,2};
  \node (n2) at (4,8)  {B,3};
  \node (n3) at (1,6)  {C,3};

    \foreach \from/\to in {n2/n1,n3/n1}
    \draw (\from) -> (\to);
\end{tikzpicture}
\end{center}

The threshold $3$ is just a large enough integer so that the colors of node $B$ or node $C$ do not change/evolve, regardless of the coloring on the graph. The system derived from such a realization is the map $f_A : 2^{\bf 3} \rightarrow 2^{\bf 3}$ satisfying A.1, A.2 and A.3 such that $BC \mapsto ABC$ and all remaining states are left unchanged.
A fixed point representation yields:
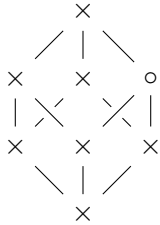
\begin{center}
\begin{tikzpicture}[scale=0.6]
  \node (max) at (0,3) {$\times$};
  \node (a) at (1.5,1.5) {$\circ$};
  \node (b) at (0,1.5) {$\times$};
  \node (c) at (-1.5,1.5) {$\times$};
  \node (d) at (1.5,0) {$\times$};
  \node (e) at (0,0) {$\times$};
  \node (f) at (-1.5,0) {$\times$};
  \node (min) at (0,-1.5) {$\times$};
  \draw (min) -- (d) -- (a) -- (max) -- (b) -- (f)
  (e) -- (min) -- (f) -- (c) -- (max)
  (d) -- (b);
  \draw[preaction={draw=white, -,line width=6pt}] (a) -- (e) -- (c);
\end{tikzpicture}
\end{center}

Similarly the maps $f_B$ and $f_C$ derived for the rules of $B$ and $C$ are represented (respectively from left to right) through their fixed points as:

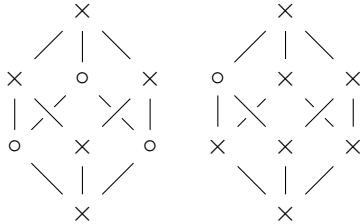
\begin{center}
\begin{tikzpicture}[scale=0.6]
  \node (max) at (0,3) {$\times$};
  \node (a) at (1.5,1.5) {$\times$};
  \node (b) at (0,1.5) {$\circ$};
  \node (c) at (-1.5,1.5) {$\times$};
  \node (d) at (1.5,0) {$\circ$};
  \node (e) at (0,0) {$\times$};
  \node (f) at (-1.5,0) {$\circ$};
  \node (min) at (0,-1.5) {$\times$};
  \draw (min) -- (d) -- (a) -- (max) -- (b) -- (f)
  (e) -- (min) -- (f) -- (c) -- (max)
  (d) -- (b);
  \draw[preaction={draw=white, -,line width=6pt}] (a) -- (e) -- (c);
  
    \node (maxmax) at (4.5,3) {$\times$};
  \node (aa) at (6,1.5) {$\times$};
  \node (bb) at (4.5,1.5) {$\times$};
  \node (cc) at (3,1.5) {$\circ$};
  \node (dd) at (6,0) {$\times$};
  \node (ee) at (4.5,0) {$\times$};
  \node (ff) at (3,0) {$\times$};
  \node (minmin) at (4.5,-1.5) {$\times$};
  \draw (minmin) -- (dd) -- (aa) -- (maxmax) -- (bb) -- (ff)
  (ee) -- (minmin) -- (ff) -- (cc) -- (maxmax)
  (dd) -- (bb);
  \draw[preaction={draw=white, -,line width=6pt}] (aa) -- (ee) -- (cc);
  
\end{tikzpicture}
\end{center}

Our overall descriptive rule of the dynamics is constructed by a \emph{descriptive} combination of the evolution rules of $A$, $B$ and $C$. With respect to the objects \emph{behind} those rules, the overall system is obtained by a $\plus$ combination of the \emph{local} systems. Indeed we have $f = f_A \plus f_B \plus f_C$, and such a combination is obtained by only keeping the fixed points that are common to all three systems.

\section{Components Realization}\label{sec:comp}

The systems derived from instances of ``models'' \emph{forget} all the componental structure described by the model. Nodes in M.0 and M.1 are bundled together to form the Boolean lattice, and the system is a monolithic map from $2^V$ to $2^V$. We have not discussed any means to recover components and interconnection structures from systems. We might want such a recovery for at least two reasons. First, we may be interested in understanding specific subparts of the modeled system. Second, we may want to realize our systems as instances of other models. In state spaces isomorphic to $2^S$ for some $S$, components may often be identified with the elements of $S$. In the case of M.0 and M.1, the components are represented as nodes in a graph. Yet, two elements of $S$ might also be tightly coupled as to form a single component. It is also less clear what the components can be in non-Boolean lattices as state spaces. We formalize such a flexibility by considering the set $\E$ of all maps $0_q \times 1_{q'}$ in $\L_Q{\times}\L_{Q'} \subseteq \L_{Q{\times}Q'}$ for $Q {\times} Q' = P$. The map $0_q \times 1_{q'}$ sends $(q,q') \in Q\times Q'$ to $(q,\hat{q}')$ where $\hat{q}'$ is the maximum element of $Q'$.  Indeed, the system $0_q$, being the identity, keeps $q$ unchanged in $Q$.  The system $1_{q'}$, being the maximum system, sends $q'$ to the maximum element $\hat{q}'$ of $Q'$. We refer to the maps of $\E$ as \emph{elementary} functions (or systems). A \emph{component realization} of $P$ is a collection of systems $e_A,\cdots,e_H$ in $\E$ where:
\begin{align}
 &e_A \plus \cdots \plus e_H = 1\nonumber\\
 &e_I \proj e_J = 0 \quad \text{for all }I\neq J\nonumber
\end{align}

For a different perspective, we consider a direct decomposition of $P$ into lattices $A,\cdots,H$ such that $A \times \cdots \times H = P$. An element $t$ of $P$ can be written either as a tuple $(t_A,\cdots,t_H)$ or as a string $t_A \cdots t_H$. If $(t_A,\cdots,t_H)$ and $(t'_A,\cdots,t'_H)$ are elements of $P$, then:
\begin{align}
  (t_A,\cdots,t_H) \join (t'_A,\cdots,t'_H) &= (t_A \join t'_A,\cdots,t_H \join t'_H)\label{eq:joinproduct}\\
  (t_A,\cdots,t_H) \meet (t'_A,\cdots,t'_H) &= (t_A \meet t'_A,\cdots,t_H \meet t'_H)\label{eq:meetproduct}.
\end{align}
Indeed, the join (resp.\ meet) in the product lattice, is the product of the joins (resp.\ meets) in the factor lattices.   Maps $e_A,\cdots,e_H$ can be defined as $e_I : ti \mapsto t\hat{i}$, that keeps $t$ unchanged and maps $i$ to the maximum element $\hat{i}$ of $I$. These maps belongs to $\L_P$, and together constitute a component realization as defined above. Conversely, each component realization gives rise to a direct decomposition of $P$.

\begin{thm}\label{thm:decom}
 Let $e_A, \cdots, e_H$ be a component realization of $P$. If $f \in \L_P$, then $f = f\proj e_A \plus \cdots \plus f\proj e_H$.
\end{thm}
\begin{proof}
 It is immediate that $f\proj e_A \plus \cdots \plus f\proj e_H \leq f$. To show the other inequality, consider $t \notin \Phi f$. Then $t_I \neq (ft)_I$ for some $I$. Furthermore, if $t' \geq t$ with $t'_I = t_I$, then $t'\notin \Phi f$. Assume $t \in \Phi (f\proj e_I)$, then $t = s \meet r$ for some $s \in \Phi f$ and $r \in \Phi e_I$. It then follows that $r_I = \hat{i}$, the maximum element of $I$. Therefore $s_I = t_I$ and $s \geq t$ contradicting the fact that $s \in \Phi f$.
\end{proof}

The map $f\proj e_I$ may evolve only the $I$-th \emph{component} of the state space.
\begin{prop}\label{pro:evolveIcomponent}
  If $s\in P$ is written as $ti$, then $(f \proj e_I)s = t(fs)_I$, where $(fs)_I$ is the projection of $fs$ onto the component $I$.
\end{prop}
\begin{proof}
 We have $(f \proj e_I)s = fs \meet e_Is = f(ti) \meet t\hat{i} = t(fs)_I$.  The last equality follows from Equation \ref{eq:meetproduct}. 
\end{proof}
It is also the evolution rule governing the state of component $I$ as a function of the full system state.

\begin{prop}\label{pro:joincomponent}
  Let $e_A, \cdots, e_H$ be a component realization of $P$. If $f \in \L_P$, then $fa = (f\proj e_A)a \vee \cdots \vee (f\proj e_H)a$ for every $a \in P$.
\end{prop}
\begin{proof}
  It is immediate that $(f\proj e_A)a \vee \cdots \vee (f\proj e_H)a \leq fa$.  The other inequality follows from combining Proposition \ref{pro:evolveIcomponent} and Equation \ref{eq:joinproduct}.
\end{proof} 
\begin{exa} Let $f$ be the system derived from an instance $(V,A)$ of M.0. We consider the maps $e_i : X \mapsto X\cup\{i\}$ for $i \in V$. The collection $\{e_i\}$ forms a \emph{component realization} where $e_i$ corresponds to node $i$ in the graph. The system $f \proj e_i$ may be identified with the \emph{ancestors} of $i$, namely, nodes $j$ where a directed path from $j$ to $i$ exists. A realization (in the form of M.0) of $f \proj e_i$ then colors $i$ $black$ whenever any ancestor of it is $black$, leaving the color of all other nodes unchanged. Combining the maps $f \proj e_i$ recovers the map $f$.
\end{exa}
  
\emph{Interconnection structures} (e.g. digraphs as used in M.1) may be further derived by defining projection and inclusion maps accordingly and requiring the systems to satisfy some fixed-point conditions. Such structures can be interpreted as systems in $\L_{\L_P}$. They will not be considered in this paper.

\subsection{Defining Cascade Effects}

Given a component realization $e_A, \cdots, e_H$, define a collection of maps $f_A,\cdots,f_H$ where $f_I \leq e_I$ dictates the evolution of the state of component $I$ as a function of $P$. These update rules are typically combined to form a system $f = f_A \plus \cdots \plus f_H$. \emph{Cascade effects} are said to occur when $f \proj e_I \neq f_I$ for some $I$. The behavior governing a certain (sub)system $I$ is \emph{enhanced} as this component is embedded into the greater system.  We should consider the definition provided, in this subsection, as conceptually illustrative rather than useful and complete. The main goal of the paper is to define a class of systems exhibiting cascade effects.  It is not to define what cascade effects are.  We instead refer the reader to \cite{ADAM:Dissertation} for an actionable definition and a study of these effects.  We will however revisit this definition in Section \ref{sec:eval} with more insight.  

The conditions under which such effects occurs depend on the properties of the operations. If $\proj$ distributes over $\plus$, then this behavior is never bound to occur; this will seldom be the case as will be shown in the next section.

\subsection{Overview Through An Example (Continued)}

We continue the running example. On a dual end, if we wish to view the nodes $A$, $B$ and $C$ as distinct entities, we may define a component realization $e_A$, $e_B$ and $e_C$ represented (respectively from left to right) as:

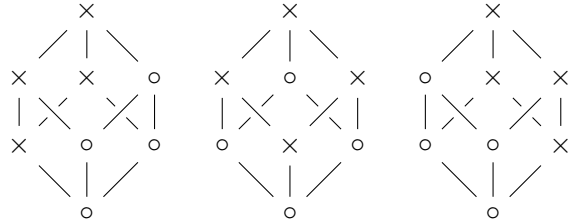
\begin{center}
\begin{tikzpicture}[scale=0.6]
  \node (max) at (0,3) {$\times$};
  \node (a) at (1.5,1.5) {$\circ$};
  \node (b) at (0,1.5) {$\times$};
  \node (c) at (-1.5,1.5) {$\times$};
  \node (d) at (1.5,0) {$\circ$};
  \node (e) at (0,0) {$\circ$};
  \node (f) at (-1.5,0) {$\times$};
  \node (min) at (0,-1.5) {$\circ$};
  \draw (min) -- (d) -- (a) -- (max) -- (b) -- (f)
  (e) -- (min) -- (f) -- (c) -- (max)
  (d) -- (b);
  \draw[preaction={draw=white, -,line width=6pt}] (a) -- (e) -- (c);
  
    \node (maxmax) at (4.5,3) {$\times$};
  \node (aa) at (6,1.5) {$\times$};
  \node (bb) at (4.5,1.5) {$\circ$};
  \node (cc) at (3,1.5) {$\times$};
  \node (dd) at (6,0) {$\circ$};
  \node (ee) at (4.5,0) {$\times$};
  \node (ff) at (3,0) {$\circ$};
  \node (minmin) at (4.5,-1.5) {$\circ$};
  \draw (minmin) -- (dd) -- (aa) -- (maxmax) -- (bb) -- (ff)
  (ee) -- (minmin) -- (ff) -- (cc) -- (maxmax)
  (dd) -- (bb);
 \draw[preaction={draw=white, -,line width=6pt}] (aa) -- (ee) -- (cc);

  \node (maxmaxmax) at (9,3) {$\times$};
  \node (aaa) at (10.5,1.5) {$\times$};
  \node (bbb) at (9,1.5) {$\times$};
  \node (ccc) at (7.5,1.5) {$\circ$};
  \node (ddd) at (10.5,0) {$\times$};
  \node (eee) at (9,0) {$\circ$};
  \node (fff) at (7.5,0) {$\circ$};
  \node (minminmin) at (9,-1.5) {$\circ$};
  \draw (minminmin) -- (ddd) -- (aaa) -- (maxmaxmax) -- (bbb) -- (fff)
  (eee) -- (minminmin) -- (fff) -- (ccc) -- (maxmaxmax)
  (ddd) -- (bbb);
 \draw[preaction={draw=white, -,line width=6pt}] (aaa) -- (eee) -- (ccc);
  
\end{tikzpicture}
\end{center}

Local evolution rules may be recovered through the systems $f\proj e_A$, $f \proj e_B$ and $f \proj e_C$. Those are likely to be different than $f_A$, $f_B$ and $f_C$ as they also take into account the \emph{effects} resulting from their combination. The systems $f\proj e_A$, $f \proj e_B$ and $f \proj e_C$ are generated by considering $\Phi f \cup \Phi e_I$ and closing this set under $\cap$. They are represented (respectively from left to right) as:

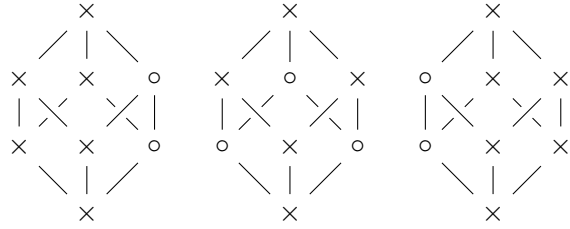
\begin{center}
\begin{tikzpicture}[scale=0.6]
  \node (max) at (0,3) {$\times$};
  \node (a) at (1.5,1.5) {$\circ$};
  \node (b) at (0,1.5) {$\times$};
  \node (c) at (-1.5,1.5) {$\times$};
  \node (d) at (1.5,0) {$\circ$};
  \node (e) at (0,0) {$\times$};
  \node (f) at (-1.5,0) {$\times$};
  \node (min) at (0,-1.5) {$\times$};
  \draw (min) -- (d) -- (a) -- (max) -- (b) -- (f)
  (e) -- (min) -- (f) -- (c) -- (max)
  (d) -- (b);
  \draw[preaction={draw=white, -,line width=6pt}] (a) -- (e) -- (c);
  
    \node (maxmax) at (4.5,3) {$\times$};
  \node (aa) at (6,1.5) {$\times$};
  \node (bb) at (4.5,1.5) {$\circ$};
  \node (cc) at (3,1.5) {$\times$};
  \node (dd) at (6,0) {$\circ$};
  \node (ee) at (4.5,0) {$\times$};
  \node (ff) at (3,0) {$\circ$};
  \node (minmin) at (4.5,-1.5) {$\times$};
  \draw (minmin) -- (dd) -- (aa) -- (maxmax) -- (bb) -- (ff)
  (ee) -- (minmin) -- (ff) -- (cc) -- (maxmax)
  (dd) -- (bb);
 \draw[preaction={draw=white, -,line width=6pt}] (aa) -- (ee) -- (cc);

  \node (maxmaxmax) at (9,3) {$\times$};
  \node (aaa) at (10.5,1.5) {$\times$};
  \node (bbb) at (9,1.5) {$\times$};
  \node (ccc) at (7.5,1.5) {$\circ$};
  \node (ddd) at (10.5,0) {$\times$};
  \node (eee) at (9,0) {$\times$};
  \node (fff) at (7.5,0) {$\circ$};
  \node (minminmin) at (9,-1.5) {$\times$};
  \draw (minminmin) -- (ddd) -- (aaa) -- (maxmaxmax) -- (bbb) -- (fff)
  (eee) -- (minminmin) -- (fff) -- (ccc) -- (maxmaxmax)
  (ddd) -- (bbb);
 \draw[preaction={draw=white, -,line width=6pt}] (aaa) -- (eee) -- (ccc);
  
\end{tikzpicture}
\end{center}

The system $f \proj e_A$ captures the fact that node $A$ can become black if only $C$ is colored black. A change in $f_A$ would, however, require both $B$ and $C$ to be black. Recombining the obtained local rules is bound to recover the overall system, and indeed $f = f\proj e_A \plus f\proj e_B \plus f\proj e_C$ as can be checked by keeping only the common fixed points.

\section{Properties of the Systems Lattice}

\emph{Complex} systems will be built out of \emph{simpler} systems through expressions involving $\plus$ and $\proj$. The power of such an expressiveness will come from the properties exhibited by the operators. Those are trivially derived from the properties of the lattice $\L$ itself.

\begin{prop}
The following propositions are equivalent. (i) The set $P$ is linearly ordered. (ii) The lattice $\L_P$ is distributive. (iii) The lattice $\L_P$ is modular.
\end{prop}
\begin{proof}
Property (ii) implies (iii) by definition. If $P$ is linearly ordered, then $\L$ is a Boolean lattice, as any subset of $P$ is closed under $\meet$. Therefore (i) implies (ii). Finally, it can be checked that $(f,g)$ is a modular pair if, and only if, $\Phi(f\proj g)= \Phi(f)\cup\Phi(g)$ i.e., $\Phi(f)\cup\Phi(g)$ is closed under $\meet$. If $\L_P$ is modular, then each pair of $f$ and $g$ is modular. In that case, each pair of states in $P$ are necessarily comparable, and so (iii) implies (i).
\end{proof}

The state spaces we are interested in are not linearly ordered. Non-distributivity is natural within the interpreted context of cascade effects, and has at least two implications. First, the decomposition of Theorem \ref{thm:decom} cannot follow from distributivity, and relies on a more subtle point. Second, cascade effects (as defined in Section \ref{sec:comp}) are bound to occur in non-trivial cases.

The loss of modularity is suggested by the asymmetry in the behavior of the operator. The $\plus$ operator corresponds to set intersection, whereas the $\proj$ operator (is less convenient) corresponds to a set union followed by a closure under $\meet$. Nevertheless, the lattice will be \emph{half} modular.

\begin{prop}
The lattice $\L_P$ is (upper) semimodular.
\end{prop}
\begin{proof}
It is enough to prove that if $f \proj g \leqc f$ and $f \proj g \leqc g$, then $f \leqc f\plus g$ and $g \leqc f\plus g$. If $f \proj g$ is covered by $f$ and $g$, then $|\Phi f - \Phi g| = |\Phi g - \Phi f| = 1$. Then necessarily $f \plus g$ covers $f$ and $g$.
\end{proof}

Semi-modularity will be fundamental in defining the $\mu$-rank of a system in Section 7. The lattice $\L$ is equivalently a graded poset, and admits a rank function $\rho$ such that $\rho(f \plus g) + \rho(f \proj g) \leq \rho(f) + \rho(g)$. The quantity $\rho(f)$ is equal to the number of non-fixed points of $f$ i.e. $|P - \Phi f|$. More properties may still be extracted, up to full characterization of the lattice. Yet, such properties are not needed in this paper.

\subsection{Additional Remarks on the Lattice of Systems}\label{sec:addRemarks}

This subsection illustrates some basic lattice theoretic properties on $2^{\bf 2}$, represented through its Hasse diagram below. We follow the notation of the running example (see e.g., Subsection \ref{sec:running}).

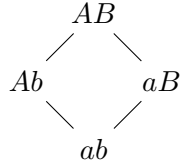
\begin{center}
\begin{tikzpicture}[scale=0.6]
  \node (max) at (0,3) {$AB$};
  \node (a) at (1.5,1.5) {$aB$};
  \node (c) at (-1.5,1.5) {$Ab$};
  \node (e) at (0,0) {$ab$};
  \draw (a) -- (max) -- (c)
  (a) -- (e) -- (c);
\end{tikzpicture}
\end{center}

The lattice $\L_{2^{\bf 2}}$ may be represented as follows. The systems are labeled through their set of fixed-points.

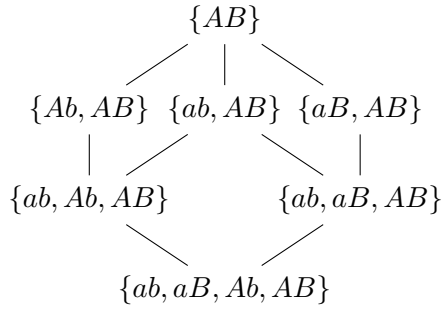
\begin{center}
\begin{tikzpicture}[scale=1.2]
  \node (max) at (0,3) {$\{AB\}$};
  \node (a) at (1.5,2) {$\{aB,AB\}$};
  \node (b) at (0,2) {$\{ab,AB\}$};
  \node (c) at (-1.5,2) {$\{Ab,AB\}$};
  \node (d) at (1.5,1) {$\{ab,aB,AB\}$};
 
  \node (f) at (-1.5,1) {$\{ab,Ab,AB\}$};
  \node (min) at (0,0) {$\{ab,aB,Ab,AB\}$};
  \draw (min) -- (d) -- (a) -- (max) -- (b) -- (f)
  (min) -- (f) -- (c) -- (max)
  (d) -- (b);
 
\end{tikzpicture}
\end{center}

A map $f \in \L_P$ will be called \emph{prime} if $P - \Phi f$ is closed under~$\meet$. Those maps will be extensively used in Section \ref{sec:failure}.

All the systems are prime (i.e. have the set of non-fixed points closed under $\cap$) except for $\{ab,AB\}$. The lattice $\L_{2^{\bf 2}}$ is (upper) semimodular as a pair of systems are covered by their join ($\plus$) whenever they cover their meet ($\proj$). All pairs form modular pairs except for the pair $\{Ab,AB\}$ and $\{aB,AB\}$. The lattice $\L_{2^{\bf 2}}$ is graded, and the (uniform) rank of a system is equal to the number of its non-fixed points as can be checked.

\subsubsection*{On Atoms and Join-irreducible elements}

An atom is an element that covers the minimal element of the lattice. In $\L_{2^{\bf 2}}$, those are $\{ab,aB,AB\}$ and $\{ab,Ab,AB\}$. A join-irreducible element is an element that cannot be written as a join of \emph{other} elements. An atom is necessarily a join-irreducible element, however the converse need not be true. The systems $\{aB,AB\}$ and $\{bA,AB\}$ are join-irreducible but are not atoms.

The join-irreducible elements in $\L_P$ may be identified with the pairs $(s,t) \in P\times P$ such that $t$ covers $s$. They can be identified with the edges in the Hasse diagram of $P$. For a covering pair $(s,t)$, define $f_{st}$ to be the least map such that $s \mapsto t$. Then $f_{st}$ is join-irreducible for each $(s,t)$, and every element of $\L_P$ is a join of elements in $\{f_{st}\}$.

\begin{prop}
 The map $f_{st}$ is prime for every $(s,t)$.
\end{prop} 

\begin{proof}
 The map $f_{st}$ is the least map such that $s \mapsto t$. It follows that $s$ is the least non-fixed point of $f_{st}$, and that every element greater than $t$ belongs to $\Phi f_{st}$. If $a, b \notin \Phi f_{st}$, then their meet $a\meet b$ is necessarily not greater than $t$, for otherwise we get $a,b \in \Phi f_{st}$. If $a\meet b$ is comparable to $t$, then $a\meet b = s \notin \Phi f_{st}$. If $a\meet b$ is non-comparable to $t$, then $(a \meet b) \meet t = s$, and so again $a \meet b \notin \Phi f_{st}$.
\end{proof}

\subsubsection*{On Coatoms and Meet-irreducible Elements}

A coatom is an element that is covered by the maximal element of the lattice. In $\L_{2^{\bf 2}}$, those are $\{ab,AB\}$, $\{aB,AB\}$ and $\{Ab,AB\}$.
In general, the coatoms of $\L$ are exactly the systems $f$ where $|\Phi f| = 2$. Note that the maximal element $\hat{p}$ of $P$ is always contained in $\Phi f$.

\begin{prop}
 Every $f \in \L_P$ is a meet of coatoms.
\end{prop}

\begin{proof}
 For each $a \in P$, let $c_a \in \L$ be such that $\Phi c_a = \{a,\hat{p}\}$. If $\Phi f = \{a,b,\cdots, h\}$, then $f = c_a \proj c_b \proj \cdots \proj c_h$. 
\end{proof} 

Such lattices are called co-atomistic. The coatoms, in this case, are the only elements that cannot be written as a meet of \emph{other} elements.

\section{On Least Fixed-Points and Cascade Effects} \label{sec:eval}

The systems are defined as maps $P \rightarrow P$ taking in an input and yielding an output.  The interaction of those systems (via the operator $\plus$) however does not depend on functional composition or application.  It is only motivated by them, and the input-ouput functional structure has been discarded throughout the analysis.  It will then also be more insightful to not view $f(a)$ as functional application.  Such a change of viewpoint can be achieved via a good use of least fixed-points.  The change of view will also lead us the a more general notion of cascade effects.

We may associate to every $a \in P$ a system $\free(a): - \mapsto - \vee a$ in $\Ell_P$.  We can then interpret $f(a)$ differently:
\begin{prop}
  The element $f(a)$ is the least fixed-point of $f + \free(a)$.
\end{prop}
\begin{proof}
  We have $f(a) = \meet \{p \in \Phi(f) : a \leq p\} = \meet \{p \in \Phi(f) \cap \Phi(\free(a))\}$. The result follows as $\Phi(f) \cap \Phi(\free(a)) = \Phi(f + \free(a))$.
\end{proof}
The map $\free : P \rightarrow \Ell_P$ is order-preserving. It also preserves joins.  Indeed, if $a, b \in P$, then $\free(a) + \free(b) = \free(a \vee b)$.  Conversely, as each map in $\Ell_P$ admits a least fixed-point, we define $\eval : \Ell_P \rightarrow P$ to be the map sending a system to its least fixed-point.  The map $\eval$ is also order-preserving, and we obtain:
\begin{thm}
  If $a \in P$ and $f \in \Ell_P$, then:
  \begin{equation*}
    \free(a) \leq f \quad \text{ if, and only if, } \quad  a \leq \eval(f)
  \end{equation*}
\end{thm}
\begin{proof}
  If $\free(a) \leq f$, then $a \leq b$ for every fixed-point $b$ of $f$.  Conversely, if $a \leq \eval(f)$, then $\{b \in P : a \leq b\}$ contains $\Phi(f)$, the set of fixed points of $f$. 
\end{proof}
The pair of maps $\free$ and $\eval$ are said to be adjoints, and form a Galois connection (see e.g., \cite{BIR1967} Ch. V, \cite{EVE1944}, \cite{ORE1944} and \cite{ERN1993} for a treatment on Galois connections).  The intuition of cascading phenomena can be seen to partly emerge from this Galois connection.  By duality, the map $\eval$ preserves meets.  Indeed, the least fixed-point of $f \proj g$ is the meet of the least fixed-points of $f$ and $g$.  The map $\eval$ does not however always preserve joins.  Such a fact causes cascading intuition to arise.  For some pairs $f, g \in \Ell_P$, we get:
\begin{equation}\label{eq:inequality}
  \eval(f + g) \neq \eval(f) \vee \eval(g)
\end{equation}
Generally, two systems interact to yield, combined, something greater than what they yield separately, then combined.  Specifically, consider $f \in \Ell_P$ and $a \in P$ such that $\eval(f) \leq a$.  If $\eval(f + \free(a)) \neq \eval(f) \vee \eval(\free(a))$, then $f(a) \neq a$.  In this case, the point $a$ \emph{expanded} under the map $f$, and cascading effects have thus occured.  The paper will not pursue this line of direction.  This direction is extensively pursued in \cite{ADAM:Dissertation}.  Also, a definition of cascade effects was already introduced in Section \ref{sec:comp}.  We thus briefly revisit it and explain the connection to the inequality presented.   The inequality can be further explained by the semimodularity of the lattice, but such a link will not be pursued.

\subsection{Revisiting Component Realization}  Given a component realization $e_A, \cdots, e_H$ of $P$, we let $f_A,\cdots,f_H$ be a collection of maps where $f_I \leq e_I$ dictates the evolution of the state of component $I$ as a function of $P$.  If $f = f_A \plus \cdots \plus f_H$, then recall from Section \ref{sec:comp} that \emph{cascade effects} are said to occur when $f \proj e_I \neq f_I$ for some $I$.

We will illustrate how this definition links to the inequality obtained from the Galois connection.  For simplicity, we consider only two components $A$ and $B$.  Let $e_A,e_B$ be a component realization of $P$, and consider two maps $f_A,f_B$ where $f_I \leq e_I$.  Define $f = f_A \plus f_B$. If $f \proj e_A \neq f_A$, then $(f \proj e_A)a \neq f_Aa = a$ for some fixed point $a$ of $f_A$.  We then have $fa \neq f_A a \join f_B a$.  As $fa =  (f \proj e_A)a \join (f \proj e_B)a$ by Proposition \ref{pro:joincomponent}, we get:
\begin{equation} \label{eq:connection}
  \eval\big( f_A \plus \free(a) \plus f_B \plus \free(a) \big) \neq \eval\big( f_A \plus \free(a)\big) \vee \eval\big( f_B \plus \free(a)\big) 
\end{equation}
Conversely, if Equation \ref{eq:connection} holds, then either $f_A a \neq (f\proj e_A) a$ or $f_B a \neq (f\proj e_B) a$.

\subsection{More on Galois Connections}

The inequality in Equation \ref{eq:inequality} gives rise to cascading phenomena in our situation.  It is induced by the Galois connection between $\free$ and $\eval$, and the fact that $\eval$ does not preserve joins.  The content of the lattices can however be changed, keeping the phenomenon intact.  Both the lattice of systems $\Ell_P$ and the lattice of states $P$ can be replaced by other lattices.  If we can setup another such inequality for the other lattices, then we would have created cascade effects in a different situation.  We refer the reader to \cite{ADAM:Dissertation} for a thorough study along those lines.  The particular class of systems studied in this paper is however somewhat special.  Indeed, every system itself arises from a Galois connection. Thus, if we focus on a particular system $f$, then we get a Galois connection induced by the inclusion:
\begin{equation*}
  \Phi(f) \rightarrow P
\end{equation*}
And indeed, cascade effects will emerge whenever $a \vee_{\Phi(f)} b \neq a \vee_P b$.  This direction will not be further discussed in the paper.  

This double presence of Galois connections seems to be merely a coincidence. It implies however that we can recover cascading phenomena in our situation at two levels:  either at the level of systems interacting or at the level of a unique system with its states interacting.

\subsection{Higher-Order Systems}

For a lattice $P$, we constructed the lattice $\Ell_P$.  By iterating the construction once, we may form $\Ell_{\Ell_P}$.  Through several iterations, we may recursively form $\Ell^{m+1}_P = \Ell_{\Ell^m_P}$ with $\Ell^0_P = P$.  Systems in $\Ell^m_P$ take into account nested if-then statements.  The construction induces a map $\eval: \Ell^{m+1}_P \rightarrow \Ell^{m}_P$, sending a system to its least fixed-point.  We thus recover a sequence:
\begin{equation*}
  \cdots \rightarrow \Ell^{3}_P \rightarrow \Ell^{2}_P \rightarrow \Ell_P \rightarrow P
\end{equation*}
The $\free$ map construction induces an inclusion $\Ell^{m}_P \rightarrow \Ell^{m+1}_P$ for every $m$.  We may then define an infinite lattice $\Ell_P^{\infty} = \bigcup^\infty_{m=1} \Ell^m_P$ that contains all finite higher-order systems.  We may also decide to complete $\Ell_P^{\infty}$ in a certain sense to take into account infinite recursion.  Such an idea have extensively recurred in denotational semantics and domain theory (see e.g., \cite{SCO1972}, \cite{SCO1972A} and \cite{SCO1972B}) to yield semantics to programming languages, notably the $\lambda$-calculus.  This idea will however not be further pursued in this paper.

\section{Connections to Formal Methods}\label{sec:formal}

The ideas developed in this paper intersect with ideas in formal methods and semantics of languages. To clarify some intersections, we revisit the axioms. A map $f: P \rightarrow P$ belongs to $\Ell_P$ if it satisfies:
\begin{description}
 \item[A.1] If $a \in P$, then $a \leq fa$.
 \item[A.2] If $a, b \in P$ and $a \leq b$, then $fa \leq fb$.
 \item[A.3] If $a \in P$, then $ffa = fa$.
\end{description}
The axiom A.2 may generally be replaced by one requiring the map to be \emph{scott-continuous}, see e.g. \cite{SCO1972B} for a definition.  Every scott-continuous function is order-preserving, and in the case of finite lattices (as assumed in this paper) the converse is true.  The axiom A.3 may then be discarded, and fixed points can generally be recovered by successive iterations of the map (ref. the Kleene fixed-point theorem). The axiom A.1 equips the systems with their expansive nature.  The more important axiom is A.2 (or potentially scott-continuity) which is adaptive to the underlying order.  Every map satisfying A.2 can be \emph{closed} into a map satisfying A.1 and A.2, by sending $f(-)$ to $- \vee f(-)$.  The least fixed-points of both coincide.

The interplay of A.1 and A.2 ensures that concurrency of update rules in the systems does not produce any conflicts.  The argument is illustrated in Proposition \ref{pro:join}, and is further fully refined in Subsection \ref{ap:motivation}.  The systems can however capture concurrency issues by considering power sets. As an example, given a Petri net, we may construct a map sending a set of initial token distribution, to the set of all possible token distributions that can be \emph{caused} by such an initial set.  This map is easily shown to satisfy the axioms A.1, A.2 and A.3.  A more elaborate interpretation of the state space, potentially along the lines of event structures as described in \cite{PLO1981}, may lead to further connections for dealing with concurrency issues.

The interplay of lattices and least fixed-point appears throughout efforts in formal methods and semantics of languages.  We illustrate the relevance of A.1 and A.2 via the simple two-line program \texttt{Prog}:
\begin{verbatim} 
1. while ( x > 5 ) do 
2. x := x - 1;
\end{verbatim}
We define a state of this program to be an element of $\Sigma := \mathbb{N} \times \{in_1,out_1,in_2,out_2\}$.  A number in $\mathbb{N}$ denotes the value assigned to \texttt{x}, and $in_i$ (resp.\ $out_i$) indicates that the program is entering (resp.\ exiting) line $i$ of the program.  For instance, $(7,out_2)$ denotes the state where \texttt{x} has value $7$ right after executing line $2$.  We define a finite execution trace of a program to be a sequence of states that can be reached by some execution of the program in finite steps.  A finite execution trace is then an element of $\Sigma^*$, the semigroup of all finite strings over the alphabet $\Sigma$. Two elements $s, s' \in \Sigma^*$ can be concatenated via $s \circ s'$.

We then define $f : 2^{\Sigma^*} \rightarrow 2^{\Sigma^*}$ such that:
\begin{align}\label{eq:trace}
  B \mapsto  f(B) := & \big\{(n,in_1) : n \in \mathbb{N} \big\}\nonumber\\
  \cup & \big\{ tr\circ(n,out_1) : tr \in B \text{ and } tr \in \Sigma^* \circ (n,in_1) \big\}\nonumber\\
  \cup & \big\{ tr\circ(n,in_2) : tr \in B \text{ and } tr \in \Sigma^* \circ(n,out_1) \text{ and } n>5  \big\}\\
  \cup & \big\{ tr\circ(n,out_2): tr \in B \text{ and } tr \in \Sigma^* \circ(n+1,in_2)\big\}\nonumber \\
  \cup & \big\{ tr\circ(n,in_1) : tr \in B \text{ and } tr \in \Sigma^* \circ(n,out_2)\big\}\nonumber
\end{align}
The map $f$ satisfies A.1 and A.2.  If $B_{sol} \subseteq \Sigma^*$ is the set of finite excution traces, then $B_{sol} \supseteq f(B_{sol})$.  Furthermore, $B_{sol}$ is the least fixed of $f$.  This idea is pervasive in obtaining semantics of programs.  The maps $f$, in deriving semantics, are however typically only considered to be order-preserving (or Scott-continuous).  The connection to using maps satisfying both A.1 and A.2 somewhat hinges on the fact that for every order-preserving map $h$, the least fixed-point of $h(-)$ and $-\vee h(-)$ coincide.  The map $f$ may also be closed under A.3 via successive iterations, without modifying the least fixed-point, to yield a map in $\Ell_{2^{\Sigma^*}}$.  We refer the reader to \cite{NIE1999} Ch 1 for an overview of various methods along the example we provide, the work on abstract interpretation (see e.g., \cite{COU1977} and \cite{COU2000}) for more details on traces and semantics, and the works \cite{SCO1972}, \cite{SCO1972A} and \cite{SCO1972B} for the relevance of A.2 (or Scott-continuity) in denotational semantics.  In a general poset, non-necessarily boolean, we recover the form of M.4.  Galois connections also appear extensively in abstract interpretation.  The methods of abstract interpretation can be enhanced and put to use in approximating (and further understanding) the systems in this paper. 

Various ideas present in this paper may be further linked to other areas.  That ought not be surprising as the axioms are very minimal and natural.  From this perspective, the goal of this work is partly to guide efforts, and very effective tools, in the formal methods community into dealing with cascade-like phenomena. 

\subsection{Cascading Phenomena in this Context}

We also illustrate cascade effects, as described in Section \ref{sec:eval}, in the context of programs.  Consider another program \texttt{Prog'}:
\begin{verbatim} 
1. while ( x is odd ) do 
2. x := x - 1;
\end{verbatim}
Each of \texttt{Prog} and \texttt{Prog'} ought to be thought of as a partial description of a \emph{larger} program.  Their interaction yields the simplest program allowing both descriptions, namely:
\begin{verbatim}
1. while ( x > 5 ) or ( x is odd ) do 
2. x := x - 1;
\end{verbatim}
Let $f$ and $g$ be the maps (satisfying A.1 and A.2) attributed to \texttt{Prog} and \texttt{Prog'} respectively, as done along the lines of Equation \ref{eq:trace}.  The set of finite execution traces of the combined program is then the least fixed-point of $f \vee g$, where $(f \vee g)B = fB \cup gB$. Note that $f \vee g$ then satisfies both A.1 and A.2. Cascade effects then appear upon interaction.  The interaction of the program descriptions is bound to produce new traces that cannot be accounted for by the traces of the separate programs.  Indeed, every trace containing:
$$(5,out_2)\circ(5,in_1)\circ(5,out_1)\circ(5,in_2)$$
allowed in the combined program is not allowed in neither of the separate programs.  Formally, define a map $\eval$ that sends a function $2^{\Sigma^*} \rightarrow 2^{\Sigma^*}$ satisfying A.1 and A.2 to its least fixed point.  The map $\eval$ is well defined as $2^{\Sigma^*}$ is a complete lattice.  We then get an inequality:
\begin{equation*}
  \eval( f \vee g) \neq \eval(f) \cup \eval(g)
\end{equation*}
We may also link back to systems in $\Ell$ and the cascade effects' definition provided for them.  If $\bar{f}$ and $\bar{g}$ denote the closure of $f$ and $g$ in $2^{\Sigma^*}$ to satisfy A.3 (e.g. via iterative composition in the case of scott-continuous functions), then the closure of $f \vee g$ corresponds to $\bar{f} \plus \bar{g}$.  Of course, for every $h$ satisfying A.1 and A.2, both $h$ and $\bar{h}$ have the same least fixed-point. We then have:
\begin{equation*}
   \eval( \bar{f} \plus \bar{g}) \neq \eval( \bar{f}) \cup \eval(\bar{g})
\end{equation*}

The paper will mostly be concerned with properties of the systems in $\Ell$. The direction of directly studying the inequality will not be pursued in the paper.  It is extensively pursued in \cite{ADAM:Dissertation}.

\section{Shocks, Failure and Resilience} \label{sec:failure}

The theory will be interpreted within cascading failure. The informal goal is to derive conditions and insight determining whether or not a system hit by a shock would fail. Such a statement requires at least three terms---\emph{hit}, \emph{shock} and \emph{fail}---to be defined.

The situation, in the case of the models M.i, may be interpreted as follows.  Some components (or agents) initially fail (or become infected).  The dynamics then lead other components (or agents) to fail (or become infected) in turn.  The goal is to assess the conditions under which a large fraction of the system's components fail.  Such a state may be reached even when a very small number of components initially fail.  This section aims to quantify and understand the resilence of the system to initial failures.  Not only may targeted componental failures be inflicted onto the system, but also external (exogenous) rules may act as shocks providing conditional failures in the systems.  A shock in this respect is to be regarded as a system.  This remark is the subject of the next subsection.

\subsection{A Notion of Shock} 

Enforcing a \emph{shock} on a system would intuitively yield an evolved system incorporating the effects of the shock. Forcing such an intuition onto the identity system leads us to consider shocks as systems themselves. Any shock $s$ is then an element of $\L_P$. Two types of shocks may further be considered. \emph{Push shocks} evolve state $\check{p}$ to some state $a$. \emph{Pull shocks} evolve some state $a$ directly to $\hat{p}$. Allowing arbitrary $\plus$ and $\proj$ combinations of such systems generates $\L$. The set of shocks is then considered to be the set $\L$.

Shocks trivially inherit all properties of systems, and can be identified with their fixed points as subsets of $P$. Finally, a shock $s$ \emph{hits} a system $f$ to yield the system $f\plus s$.

\begin{exa} One example of shocks (realized through the form of M.i) inserts element to the initial set $X_0$ to obtain $X_0'$. This shock corresponds to the (least) map in $\L$ that sends $\emptyset$ to $X_0'$. Equivalently, this shock has as a set of fixed points the principal (upper) order filter of the lattice $P$ generated by the set $X_0'$ (i.e. the fixed points are all, and only, the sets containing $X_0'$). Further shocks may be identified with decreasing $k_i$ or adding an element $j$ to $N_i$ for some $i$.
\end{exa}
  
\subsubsection*{Remark} It will often be required to restrict the space of shocks. There is no particular reason to do so now, as any shock can be well justified, for instance, in the setting of M.3. We may further wish to keep the generality to preserve symmetry in the problem, just as we are not restricting the set of systems.

\subsection{A Notion of Failure}

A shock is considered to fail a system if the mechanisms of the shock combined with those of the system evolve the most desirable state to the least desirable state. Shock $s$ fails system $f$ if, and only if, $s+f = 1$.

In the context of M.i, failure occurs when $X_{|S|}$ contains all the elements of $S$. This notion of failure is not restrictive as it can simulate other notions. As an example, for $C \subseteq P$, define $u_C \in \L$ to be the least system that maps $a$ to $\hat{p}$ if $a\in C$. Suppose shock $s$ ``fails'' $f$ if $(f+s) a \geq c$ for some $c \in C$ and all $a$. Then $s$ ``fails'' $f$ if, and only if, $f + s + u_c = 1$. The notion may further simulate notions of failure arising from monotone propositional sentences. If we suppose that $(s_1,s_2,s_3)$ ``fails'' $(f_1,f_2,f_3)$ if ($s_1$ fails $f_1$) and (either $s_2$ fails $f_2$ or $s_3$ fails $f_3$), then there is a map $\psi$ into $\L$ such that $(s_1,s_2,s_3)$ ``fails'' $(f_1,f_2,f_3)$ if, and only if, $\psi(s_1,s_2,s_3) \plus \psi(f_1,f_2,f_3) = 1$. We can generally construct a monomorphism $\psi : \L_P {\times} \L_Q \rightarrow \L_{P{\times}Q}$ such that $s \plus f = 1$ and (or) $t \plus g = 1$ if, and only if, $\psi(s,t) \plus \psi(f,g) =~1$.

\subsection{Minimal Shocks and Weaknesses of Systems}

We set to understand the class of shocks that fail a system. We define the collection $\S_f$:
\begin{equation*}
 \S_f = \{ s \in \L : f + s = 1\}
\end{equation*}
 As a direct consequence of Theorem \ref{thm:iso}, we get:
\begin{cor}\label{pro:failcond}
 Shock $s$ belongs to $S_f$ if, and only if, $\Phi f \cap\Phi s = \{\hat{p}\}$
\end{cor}
For instances of M.i, it is often a question as to whether or not there is some $X_0$ with at most $k$ elements, where the final set $X_{|S|}$ contains all the elements of $S$. Such a set exists if, and only if, for some set $X$ of size $k$, all sets containing it are non-fixed points (with the exception of $S$).

If $s \leq s'$ and $s \in S_f$, then $s' \in S_f$. Thus, an understanding of $\S_f$ may come from an understanding of its minimal elements. We then focus on the \emph{minimal shocks} that fail a system $f$, and denote the set of those shocks by $\check{S}_f$:
\begin{equation*}
 \check{\S}_f = \{ s \in \S_f : \text{for all }t \in S_f, \text{ if }t \leq s \text{ then } t=s\}
\end{equation*}

A map $f \in \L_P$ will be called \emph{prime} if $P - \Phi f$ is closed under~$\meet$. A prime map $f$ is naturally complemented in the lattice, and we define $\neg f$ to be (the prime map) such that $\Phi (\neg f) = P - (\Phi f - \{\hat{p}\})$. If $f$ is prime, then $\neg \neg f = f$.

\begin{prop}
 The system $f$ admits a unique minimal shock that fails it, i.e. $|\check{\S}_f| = 1$  if, and only if, $f$ is prime.
\end{prop}
\begin{proof}
 If $f$ is prime, then $\neg f \in \S_f$. The map $\neg f$ is also the unique minimal shock as if $s \in \S_f$, then $\Phi s \subseteq \Phi \neg f$ by Proposition \ref{pro:failcond}. Conversely, suppose $f$ is not prime. Then $a = b \meet c$ for some $a \in \Phi f$ and $b, c \notin \Phi f$. Define $b' = fb$ and $c' = fc$ and consider the least shocks $s_0, s_{b'}$ and $s_{c'}$ such that $s_0\check{p} = a, s_{b'} b' = \hat{p}$ and $s_{c'} c' = \hat{p}$. Furthermore, define $s_b$ and $s_c$ such that $s_b a = b$ and $s_c a = c$. Then $b \in \Phi s_b$ and $c \in \Phi s_c$. Finally, $s_0 + s_b + s_{b'}$ and $s_0 + s_c + s_{c'}$ belong to $\S_f$, but their meet is not in $\S_f$ as $a$ is a fixed point of $(s_0 + s_b + s_{b'}) \proj (s_0 + s_c + s_{c'})$. This contradicts the existence of a minimal element in $\S_f$. 
\end{proof}

As an example, consider an instance of M.1 where ``the underlying graph is undirected'' i.e. $i \in N_j$ if, and only if, $j \in N_i$. Define $f$ to be the map $X_0 \mapsto X_{|S|}$. If $f(\emptyset) = \emptyset$ and $f(S-\{i\}) = S$ for all $i$, then $|\check{S}_f|\neq 1$ i.e. there are at least two minimal shock that fail $f$. Indeed, consider a minimal set $X$ such that $fX \neq X$. If $Y = (X \cup N_i) - \{i\}$ for some $i \in X$, then $fY \neq Y$. However, $f (X\cap Y) = X \cap Y$ by minimality of~$X$.

\begin{thm} 
 If $s$ belongs to $\check{\S}_f$, then $s$ is prime.
\end{thm}
\begin{proof}
 Suppose $s$ is not prime. Then, there exists a minimal element $a = b \meet c$ such that $a \in \Phi s$ and $b,c \notin \Phi s$. We consider $(b,c)$ to be \emph{minimal} in the sense that for $(b',c')\neq (b,c)$, if $b' \meet c' = a$, $b'\leq b$ and $c' \leq c$ then either $b' \in \Phi s$ or $c' \in \Phi s$. As $a \in \Phi s$ and $s \in \S_f$, it follows that $a \notin \Phi f$. Therefore, at least one of $b$ or $c$ is not in $\Phi f$. Without loss of generality, suppose that $b \notin \Phi f$. If for each $x \in \Phi s$ non-comparable to $b$, we show that $b \meet x \in \Phi s$, then it would follow that $s$ is not minimal as $\Phi s \cup \{b\}$ is closed under $\meet$ and would constitute a shock $s' \leq s$ that fails $f$. Consider $x \in \Phi s$, and suppose $b \meet x \notin \Phi s$. If $a \leq x$, then we get $(b \meet x) \meet c = a$ contradicting the minimality of $(b,c)$. If $a$ and $x$ are not comparable, then $a \meet x \neq a$. But $a \meet x \in \Phi s$ and $a \meet x = (b \meet x) \meet c$ with both $(b \meet x)$ and $c$ not in $\Phi s$, contradicting the minimality of $a$.
\end{proof}

Dually, we define the set of prime systems \emph{contained} in $f$.
\begin{equation*}
  \W_f = \{ w \leq f : w \text{ is prime}\}
\end{equation*}
\begin{prop}
 If $f \in \L$ and $\W_f = \{w_1, \cdots, w_m\}$, then $f = w_1 \plus \cdots \plus w_m$. \nonumber
\end{prop}
\begin{proof}
 All join-irreducible elements of $\L$ are prime (see Subsection \ref{sec:addRemarks}). Therefore $\W_f$ contains all join-irreducible elements less than $f$, and $f$ is necessarily the join of those elements.
\end{proof}
Keeping only the maximal elements of $\W_f$ is enough to reconstruct $f$. We define:
\begin{equation*}
 \hat{\W}_f = \{ w \in \W_f : \text{for all }v \in \W_f, \text{ if }w \leq v \text{ then } v=w\}
\end{equation*}
\begin{prop}
 The operator $\neg$ maps $\check{\S}_f$ to~$\hat{\W}_f$ bijectively.
\end{prop}
\begin{proof}
 If $f$ is prime, then $\neg\neg f = f$. It is therefore enough to show that if $s \in \check{\S}_f$, then $\neg s \in \hat{\W}_f$ and that if $w \in \hat{\W}_f$, then $\neg w \in \check{\S}_f$. For each $s \in \check{\S}_f$, as $\neg s \leq f$, there is a $w \in \hat{\W_f}$ such that $\neg s \leq w$. Then $\neg w \leq s$, and so $s = \neg w$ as $s$ is minimal. By symmetry we get the result.
\end{proof}

We will term prime functions in $\W_f$ as \emph{weaknesses} of $f$. Every system can be decomposed injectively into its maximal weaknesses, and to each of those weaknesses corresponds a unique minimal shock that leads a system to failure. A minimal shock fails a system because it complements one maximal weakness of the system. Furthermore, whenever an arbitrary shock $s$ fails $f$ that is because a prime subshock $s'$ of $s$ complements a weakness $w$ in $f$.

\subsection{\texorpdfstring{$\mu$-}Rank, Resilience and Fragility}

We may wish to quantify the \emph{resilience} of a system. One interpretation of it may be the minimal amount of \emph{effort} required to fail a system. The word \emph{effort} presupposes a mapping that assigns to each shock some magnitude (or energy). As shocks are systems, such a mapping should coincide with one on systems.

Let $\Rnneg$ denote the non-negative reals. We expect a notion of magnitude $r : \L \rightarrow \Rnneg$ on the systems to satisfy two properties.
\begin{description}
 \item[R.1] $r(f) \leq r(g)$ if $f \leq g$ 
 \item[R.2] $r(f \plus g) = r(f) + r(g) - r(f \proj g)$ if $(f,g)$ are modular.
\end{description}
The less desirable a system is, the higher the magnitude the system has. It is helpful to informally think of a modular pair $(f,g)$ as a pair of systems that do not \emph{interfere} with each other. In such a setting, the magnitude of the combined system adds up those of the subsystems and removes that of the common part.

The rank function $\rho$ of $\L$ necessarily satisfies R.1 and R.2 as $\L$ is semi-modular. It can also be checked that, for any additive map $\mu : 2^P \rightarrow \Rnneg$, the map $f \mapsto \mu(P - \Phi f)$ satisfies the two properties. Thus, measures $\mu$ on $P$ can prove to be a useful source for maps capturing magnitude. However, any notion of magnitude satisfying R.1 and R.2 is necessarily induced by a measure on the state space.
\begin{thm}
 Let $r$ be a map satisfying R.1 and R.2, then there exists an additive map $\mu : 2^P \rightarrow \Rnneg$ such that $r(f) = \mu(P - \Phi f) + r(0)$.
\end{thm}
\begin{proof}
 A co-atom in $\L$ is an element covered by the system $1$. For each $f$, there is a sequence of co-atoms $c_1,\cdots,c_m \in \L$ such that if $f_i = c_1 \proj \cdots \proj c_i$, then $(f_i,c_{i+1})$ is a modular pair, $f_i \plus c_{i+1} = 1$ and $f_m = f$. It then follows by R.2 that $r(f_{i} \plus c_{i+1}) = r(f_i) + r(c_{i+1}) - r(f_i \proj c_{i+1})$. Therefore $r(f) = r(1) - \sum_{i=1}^m r(1) - r(c_i)$. Let $c_a$ be the co-atom with $a \in \Phi c_a$, and define $\mu(\{a\}) = r(1) - r(c_a)$ and $\mu(\{\hat{p}\}) = 0$. It follows that $r(0) = r(1) - \mu(P)$ and so $r(f) = r(0) + \mu(P) - \mu(\Phi f)$. Equivalently $r(f) = \mu(P - \Phi f) + r(0)$.
\end{proof}

As it is natural to provide the identity system $0$ with a zero magnitude, we consider only maps $r$ additionally satisfying:
\begin{description}
 \item[R.3] $r(0) = 0.$
\end{description}

Let $r$ be a map satisfying R.1, R.2 and R.3 induced by the measure $\mu$. If $\mu S = |S|$, then $r$ is simply the rank function $\rho$ of $\L$. We thus term $r$ (for a general $\mu$) as a $\mu$-rank on $\L$. The notion of $\mu$-rank is \emph{similar} to that of a norm as defined on Banach spaces. Scalar multiplication is not defined in this setting, and does not translate (directly) to the algebra presented here. However, the $\mu$-rank does give rise to a metric on $\L$.

\begin{exa} Let $f$ be the system derived from an instance $(V,A)$ in M.0, and let $\mu$ be the counting measure on $2^V$ i.e. $\mu S = |S|$. If $A$ is symmetric, then the system $f$ has $2^c$ fixed points where $c$ is the number of connected components in the graph. The $\mu$-rank of $f$ is then $2^{|V|} - 2^c$.
\end{exa}
  
Let $r$ be a $\mu$-rank. The quantity we wish to understand (termed \emph{resilience}) would be formalized as follows:
\begin{equation*}
 \resilience(f) = \min_{s \in \S_f} r(s)
\end{equation*}
We may dually define the following notion (termed \emph{fragility}):
\begin{equation*}
 \fragility(f) = \max_{w \in \W_f} r(w)
\end{equation*}
\begin{prop}
 We have $\fragility (f) + \resilience(f) = r(1)$
\end{prop}
\begin{proof}
 We have $\min_{s \in \check{S}_f} r(s) = \min_{w \in \hat{W}_f} r(\neg w)$ and $r(\neg w) = r(1) - r(w)$ for $w \in \hat{W_f}$.
\end{proof}
\begin{exa} Let $f$ be the system derived from an instance $(V,A)$ in M.0, and let $\mu$ be the counting measure on $2^V$ i.e. $\mu S = |S|$. If $A$ is symmetric, then the resilience/fragility of $f$ is tied to the size of the largest connected component of the graph. Let us define $n = |V|$. If $(V,A)$ had one component, then $\resilience(f) = 2^{n-1}$. If $(V,A)$ had $m$ components of sizes $c_1 \geq \cdots \geq c_m$, then $\resilience(f) =  2^{n - 1} + 2^{n - c_1 - 1} + \cdots + 2^{n - (c_1 + \cdots + c_{m-1}) - 1}$. As $r(1) = 2^n - 1$, it follows that $\fragility(f) = 2^n - 1 - \resilience(f)$.
\end{exa}
  
The quantity we wish to understand may be either one of $\resilience$ or $\fragility$. However, the dual definition $\fragility$ puts the quantity of interest on a comparable ground with the $\mu$-rank of a system. It is always the case that $\fragility(f) \leq r(f)$. Furthermore, equality is not met unless the system is prime. It becomes essential to quantify the inequality gap. Fragility arises only from a certain \emph{alignment} of the non-fixed points of the systems, formalized through the \emph{prime} property. Not all high ranked systems are fragile, and combining systems need not result in fragile systems although rank is increased. It is then a question as to whether or not it is possible to combine resilient systems to yield a fragile systems. To give insight into such a question, we note the following:
\begin{prop}
If $w \in \W_{f \plus g}$, then $w \leq u \plus v$ for some $u \in \W_f$ and $v \in \W_g$.
\end{prop}
\begin{proof}
 As $\neg w \plus f \plus g = 1$, it follows that $f \in \S_{\neg w \plus g}$. Then there is a $u \leq f$ in $\check{\S}_{\neg w \plus g}$. As $\neg w \plus u \plus g = 1$, it follows that $g \in {\S}_{\neg w \plus u}$. Then is a $v \leq f$ in $\check{\S}_{\neg w \plus u}$. Finally, we have $\neg w \plus u \plus v = 1$, therefore $w \leq u \plus v$.
\end{proof}
Thus a weakness can only form when combining systems through a combination of weaknesses in the systems. The implication is as follows:
\begin{cor}\label{cor:fragilityBound}
We have $\fragility(f \plus g) \leq \fragility(f) + \fragility(g)$.
\end{cor}
\begin{proof}
 For every $w \in \W_{f\plus g}$, we have $r(w) \leq max_{(u,v) \in W_{f}{\times}\W_g} r(u) + r(v)$ as $w \leq u \plus v$ for some $u \in W_f$ and $v \in W_g$. 
\end{proof}

It is not possible to combine two systems with low fragility and obtain a system with a significantly higher fragility. Furthermore, we are interested in the gap $r(f \plus g) - \fragility(f \plus g)$. If $\fragility(f) \geq \fragility(g)$, then $r(f) - 2\fragility(f)$ is a lower bound on the gap. One should be careful as such a lowerbound may be trivial in some cases. If $P$ is linearly ordered, then $\fragility(f) = r(f)$ for all $f$. The bound in this case is negative. However, if $P$ is a Boolean lattice and $\mu S = |S|$, then $r(f) - \fragility(f)$ may be in the order of $|P| = r(1)$ with $\fragility(f) \leq 2^{(-\log|P|)/2} r(f)$.

Other notions of resilience (eq. fragility) may be introduced. One notion can consider a convex combination of the $\mu$-rank of the $k$ highest-ranked shocks failing a system. The notion introduced in the paper primarily serves to illustrate the type of insight our approach might yield. Any function on the minimal shocks (failing a system) is bound to translate to a dual function on weaknesses.

\subsubsection*{Remark}    The statement of Corollary \ref{cor:fragilityBound} may be perceived to be counterintuitive.  This may be especially true in the context of cascading failure.   The statement however should not be seen to indicate that the axioms defining a system and the dynamics preclude interesting phenomena.  Indeed, it is the definition of fragility (and specifically the choice of the set of shocks over which we maximize) that gives rise to such a statement.  The statement does not imply that fragility does not emerge from the combination of resilient systems, but only that we have a bound on how much fragility increases through combinations.  The statement should not also diminish the validity of the definition of fragility, as it naturally arises from the mathematical structure of the problem.  Another, potentially more intuitive, statement on \emph{fragility} may however be recovered by a modification of the notion of fragility (or dually the notion of resilience) as follows.

We have considered so far every system to be a possible shock.  Variations on the notion of resilience may be obtained by restricting the set of possible shocks.  For instance, let us suppose that only systems of the form $s_a: p \mapsto p \vee a$ with $a\in P$ are possible shocks.  In the case of boolean lattices, these shocks can be interpreted as initially marking a subset of components (or agents) as failed (or infected).  These systems correspond, via their set of fixed-points, to the principal upper order filters of the lattice $P$.  The notion of resilience then relates to the minimum number of initial failures (on the level of components) that lead to the failure of the whole system (i.e., all components).  It is then rarely the case that two resilient systems when combined yield a resilient system.  Indeed, if $a \vee b = \hat{p}$ with $a$ and $b$ distinct from $\hat{p}$, the maximum element of $P$, then both $s_a$ and $s_b$ have some resilence. The system $s_a \plus s_b$ has however no resilience at all, as it maps every $p$ to the maximum element $\hat{p}$.

The space of possible shocks may be modified, changing the precise definition of fragility and yielding different statements.  In case there are no restrictions on shocks, we obtain Corollary \ref{cor:fragilityBound}.  We do not restrict shocks in the paper, as a first analysis, due to the lack of a  good reason to destroy symmetry between shocks and systems.  The non-restriction allows us to capture the notion of a prime system and attain a characterization of fragility in terms of maximal weaknesses.

\subsection{Overview Through An Example (Continued)}

We continue the running example. The maximal weaknesses of the system $f$ are the maximal subsystems of $f$ where the set of non-fixed points is closed under $\cap$. The system $f$ has two maximal weaknesses, represented as:

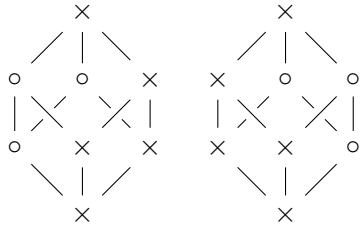
\begin{center}
\begin{tikzpicture}[scale=0.6]
  \node (max) at (0,3) {$\times$};
  \node (a) at (1.5,1.5) {$\times$};
  \node (b) at (0,1.5) {$\circ$};
  \node (c) at (-1.5,1.5) {$\circ$};
  \node (d) at (1.5,0) {$\times$};
  \node (e) at (0,0) {$\times$};
  \node (f) at (-1.5,0) {$\circ$};
  \node (min) at (0,-1.5) {$\times$};
  \draw (min) -- (d) -- (a) -- (max) -- (b) -- (f)
  (e) -- (min) -- (f) -- (c) -- (max)
  (d) -- (b);
  \draw[preaction={draw=white, -,line width=6pt}] (a) -- (e) -- (c);
  
    \node (maxmax) at (4.5,3) {$\times$};
  \node (aa) at (6,1.5) {$\circ$};
  \node (bb) at (4.5,1.5) {$\circ$};
  \node (cc) at (3,1.5) {$\times$};
  \node (dd) at (6,0) {$\circ$};
  \node (ee) at (4.5,0) {$\times$};
  \node (ff) at (3,0) {$\times$};
  \node (minmin) at (4.5,-1.5) {$\times$};
  \draw (minmin) -- (dd) -- (aa) -- (maxmax) -- (bb) -- (ff)
  (ee) -- (minmin) -- (ff) -- (cc) -- (maxmax)
  (dd) -- (bb);
  \draw[preaction={draw=white, -,line width=6pt}] (aa) -- (ee) -- (cc);
  
\end{tikzpicture}
\end{center}

The left (resp. right) weakness corresponds to the system failing when A (resp. C) is colored black. The left weakness is the map where $A \mapsto ABC$ leaving remaining states unchanged; the right weakness is the map where $C \mapsto ABC$ leaving remaining states unchanged. The system $f$ then admits two corresponding minimal shocks that fail it. Those are complements to the weaknesses in the lattice.
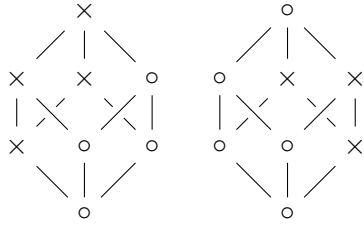
\begin{center}
\begin{tikzpicture}[scale=0.6]
  \node (max) at (0,3) {$\times$};
  \node (a) at (1.5,1.5) {$\circ$};
  \node (b) at (0,1.5) {$\times$};
  \node (c) at (-1.5,1.5) {$\times$};
  \node (d) at (1.5,0) {$\circ$};
  \node (e) at (0,0) {$\circ$};
  \node (f) at (-1.5,0) {$\times$};
  \node (min) at (0,-1.5) {$\circ$};
  \draw (min) -- (d) -- (a) -- (max) -- (b) -- (f)
  (e) -- (min) -- (f) -- (c) -- (max)
  (d) -- (b);
  \draw[preaction={draw=white, -,line width=6pt}] (a) -- (e) -- (c);
  
    \node (maxmax) at (4.5,3) {$\circ$};
  \node (aa) at (6,1.5) {$\times$};
  \node (bb) at (4.5,1.5) {$\times$};
  \node (cc) at (3,1.5) {$\circ$};
  \node (dd) at (6,0) {$\times$};
  \node (ee) at (4.5,0) {$\circ$};
  \node (ff) at (3,0) {$\circ$};
  \node (minmin) at (4.5,-1.5) {$\circ$};
  \draw (minmin) -- (dd) -- (aa) -- (maxmax) -- (bb) -- (ff)
  (ee) -- (minmin) -- (ff) -- (cc) -- (maxmax)
  (dd) -- (bb);
  \draw[preaction={draw=white, -,line width=6pt}] (aa) -- (ee) -- (cc);
  
\end{tikzpicture}
\end{center}

The left (resp. right) minimal shock can be interpreted as initially coloring node $A$ (resp. node $C$) black.

 For a counting measure $\mu$, the $\mu$-rank of $f$ is $5$, whereas the fragility of $f$ is $3$. The resilience of $f$ in that case is $4$. For a system with non-trivial rules on the components, the lowest value of fragility attainable is $1$. It is attained when all the nodes have a threshold of $2$. The highest value attainable, however, is actually $3$. Indeed, the system would have required the same amount of effort to fail it if all thresholds where equal to $1$. Yet changing all the thresholds to $1$ would necessarily increase the $\mu$-rank to $6$.

\subsection{Recovery Mechanisms and Kernel Operators}

Cascade effects, in this paper, have been mainly driven by the axioms A.1 and A.2.  The axiom A.1 ensures that the dynamics do not permit recovery.  Those axioms however do not hinder us from considering situations where certain forms of recovery are permitted, e.g., when fault-protection mechanisms are built into the systems.  Such situations may be achieved by dualizing A.1, and by considering multiple maps to define our fault-protected system.  Specifically, we define a recovery mechanism $k$ to be map $k : P \rightarrow P$ satisfying:
\begin{description}
 \item[K.1] If $a \in P$, then $ka \leq a$.
 \item[A.2] If $a, b \in P$ and $a \leq b$, then $ka \leq kb$.
 \item[A.3] If $a \in P$, then $kka = ka$.
\end{description}
The axiom K.1 is derived from A.1 by only reversing the order.  As such, a recovery mechanism $k$ on $P$ is only a system on the dual lattice $P^{\dual}$, obtained by reversing the partial order.  The maps satisfying K.1, A.2 and A.3 are typically known as \emph{kernel operators}, and inherit (by duality) all the properties of the systems described in this paper.

We may then envision a system equipped with fault-protection mechanisms as a pair $(k,f)$ where $f$ is system in $\L_P$ and $k$ is a recovery mechanism, i.e., a system in $\L_{P^{\dual}}$.  The pair $(k,f)$ is then interpreted as follows.  An initial state of failure is inflicted onto the system.  Let $a \in P$ be the initial state.  Recovery first occurs via the dynamics of $k$ to yield a more desireable state $k(a)$. The \emph{dynamics} of $f$ then come into play to yield a state $f(k(a))$.

The collection of pairs $(k,f)$ thus introduce a new class of systems, whose properties build on those developed in this paper.  If the axiom A.3 is discarded, iteration of maps in the form $(fk)^n$ may provide a more realistic account of the interplay of failures and recovery mechanisms.  In general, the map $fk$ will satisfy neither A.1 nor K.1.  A different type of analysis might thus be involved to understand these new system.

Several questions may be posed in such a setting.  For a design-question example, let us consider $P$ to be a graded poset.  What is the recovery mechanism $k$ of minimum $\mu$-rank, whereby $f(k(a))$ has rank (in $P$) less than $l$ for every $a \in P$ with rank less than $l'$?  Other design or analysis questions may posed, inspired by the example question.  This direction of recovery however will not be further investigated in this paper.

\subsubsection*{Remark} Another form of recovery may be achieved by \emph{removing rules} from the system.  Such a form may be acheived via the $\proj$ operator.  Indeed, the system $f \proj g$ is the most undesireable system that includes the common rules of both $f$ and $g$.  If $g$ is viewed as a certain complement of some system we want to remove from $f$, then we recover the required setting of recovery.  The notion of complement systems is well-defined for prime systems.  For systems that are not prime, it may be achieved by complementing the set of fixed-points, adding the maximum element $\hat{p}$ and then closing the obtained set under meets.

\section{Concluding Remarks}

Finiteness is not necessary (as explained in Section 3) for the development. The axioms A.1, A.2 and A.3 can be satisfied when $P$ is an infinite lattice, and $\Phi f$ (for every $f$) is complete whenever $P$ is complete. Nevertheless, the notion $\mu$-rank should be \emph{augmented} accordingly, and non-finite component realizations should be allowed. Furthermore, \emph{semimodularity} on infinite lattices (still holds, yet) requires stronger conditions than what is presented in this paper on finite lattices. 

Finally, the choice of the state space and order relation allows a good flexibility in the modeling exercise. State spaces may be augmented accordingly to capture desired instances. But order-preserveness is intrinsic to what is developed. This said, hints of negation (at first sight) might prove not to be integrable in this framework.

\bibliographystyle{alpha}
\bibliography{MyBiblio}

\end{document}